\documentclass[3p,10pt,a4paper,twoside,fleqn]{elsarticle}
\usepackage{amssymb,amsmath,latexsym}
\usepackage{pst-all}
\usepackage[varg]{pxfonts}
\usepackage{mathrsfs}
\newtheorem{theorem}{\bf Theorem}[section]

\newtheorem{remark}[theorem]{\bf Remark}

\newtheorem{example}[theorem]{\bf Example}

\newtheorem{algorithm}[theorem]{\bf Algorithm}
\def\endproof{\qed\endtrivlist}
\expandafter\let\csname endproof*\endcsname=\endproof

\def\qedsymbol{\ifmmode\bgroup\else$\bgroup\aftergroup$\fi
  \vcenter{\hrule\hbox{\vrule height.6em\kern.6em\vrule}\hrule}\egroup}
\def\qed{\ifmmode\else\unskip\nobreak\fi\quad\qedsymbol}

\renewcommand{\iff}{\Leftrightarrow}

\newcommand{\lra}{\leftrightarrow}
\renewcommand{\le}{\leqslant}

\newcommand{\lBrack}{\lbrack\!\lbrack}
\newcommand{\rBrack}{\rbrack\!\rbrack}

\newcommand{\oobslash}{\circledbslash\kern-10.2pt\bigcirc}
\newcommand{\ooslash}{\oslash\kern-10.2pt\bigcirc}

\begin{document}

\journal{Fuzzy Sets and Systems}

\title{\Large\bf Further improvements of determinization methods\\ for fuzzy finite automata\tnoteref{t1}}
%%%
\tnotetext[t1]{Research supported by Ministry of Education, Science and Technological Development, Republic of Serbia, Grant No. 174013}
%%%
%%%%%%%%%%%%%%%%%%%%%%%%%%%%%%%%%%%%%%%%%%%%%%%%%%
\author{Zorana Jan\v ci\' c}
\ead{zoranajancic329@gmail.com}

\author{Ivana Mici\'c}
\ead{ivanajancic84@gmail.com}

\author{Jelena Ignjatovi\'c}
\ead{jelena.ignjatovic@pmf.edu.rs}

\author{Miroslav \'Ciri\'c\corref{cor}}
\ead{miroslav.ciric@pmf.edu.rs}

\cortext[cor]{Corresponding author. Tel.: +38118224492; fax: +38118533014.}
\address{University of Ni\v s, Faculty of Sciences and Mathematics, Vi\v segradska 33, 18000 Ni\v s, Serbia}
%%%%%%%%%%%%%%%%%%%%%%%%%%%%%%%%%%%%%%%%%%%%%%%%%%

\begin{abstract}
In this paper we provide further improvements of determinization methods
for fuzzy finite automata.~These methods~perform better than all~previous determinization methods for fuzzy finite automata, developed~by B\v elohl\'avek
\cite{Bel.02b}, Li and Pedrycz \cite{LP.05}, Ignjatovi\'c et al. \cite{ICB.08}, and Jan\v ci\'c et al. \cite{JIC.11}, in the sense that they produce smaller automata, while require the same computation time.~The~only exception~is the Brzozowski type~determinization algorithm developed recently by Jan\v ci\'c and \'Ciri\'c \cite{JC.13},~which~produces a minimal crisp-deterministic~fuzzy~automaton,~but the~algorithms created here~can also be~used within the Brzozowski type algorithm and improve its performance.
\end{abstract}

\begin{keyword}
Fuzzy automaton; Crisp-deterministic fuzzy automaton; Nerode automaton; Determinization; State reduction; Complete residuated lattice;
\end{keyword}

\maketitle

\section{Introduction}

Many practical applications of automata require determinization, a procedure of converting a nondeterministic finite automaton to an equivalent deterministic finite automaton, or, in the case of fuzzy~automa\-ta, a procedure of converting a fuzzy finite automaton to an equivalent crisp-deterministic fuzzy~automaton. The standard determinization method is the subset construction, where a nondeterministic automaton with $n$ states is converted to an equivalent deterministic automaton with up to $2^n$ states, whereas in the case of fuzzy finite automata the resulting crisp-deter\-ministic fuzzy automaton can even be infinite.~That~is why the main research directions in this area are aimed at finding such methods which will mitigate the potential enormous growth of the number of states during the determinization.~The~natural~idea~is~to~com\-bine the existing determinization and state reduction methods so that we reduce the number of states to determinization.~However, here we com\-bine these methods to provide two-in-one pro\-cedures that perform determinization and state reduction simultaneously.

A crisp-deterministic fuzzy automaton is a fuzzy automaton with exactly one crisp initial~state~and~a~de\-ter\-ministic transition function, and the fuzziness is entirely concentrated in the fuzzy set of terminal~states. This kind of determinism was first studied by B\v elohl\'avek \cite{Bel.02b}, in the context of fuzzy finite automata~over a complete distributive lattice, and Li and Pedrycz \cite{LP.05}, in the context of fuzzy finite automata over a lattice-ordered monoid. Determinization algorithms that were provided there generalize the subset construction. Another algorithm, provided by Ignjatovi\'c et al.~\cite{ICB.08}, also generalizes the subset construction,~and~for~any input it generates a smaller crisp-deterministic fuzzy automaton than the algorithms from \cite{Bel.02b,LP.05}.~Since this crisp-deterministic fuzzy automaton can be alternatively constructed by means of the Nerode right con\-gruence of the original fuzzy finite automaton, it was called in \cite{ICBP.10} the {\it Nerode~automaton\/} of this~fuzzy finite auto\-maton.~The Nerode automaton was constructed in \cite{ICB.08} for fuzzy finite automata over a complete re\-sid\-uated lattice, and it was noted that the identical construction can also be applied in a more general context, for fuzzy finite automata over a lattice-ordered monoid, and even for weighted finite automata over a semi\-ring.~This construction was also transferred in \cite{CDIV.10} to weighted automata over~strong~bimonoids.\break The algorithm proposed by Jan\v ci\'c et al.~in \cite{JIC.11} produces a crisp-deterministic fuzzy automaton that~is~even smaller than the Nerode automaton.~In the terminology introduced in this paper, Jan\v ci\'c et al.~constructed the children automaton for the Nerode automaton of a given fuzzy finite automaton. Recently, Jan\v ci\'c~and \'Ciri\'c \cite{JC.13} adapted the well-known Brzozowski's double reversal determinization algorithm to fuzzy~auto\-mata.~As in the case of ordinary nondeterministic automata, Brzozowski type determinization~of~a~fuzzy finite automaton results in a minimal crisp-deterministic fuzzy automaton that is equivalent to the original fuzzy finite automaton.~It was also shown that even if all previous determinization algorithms fail to build a finite crisp-deterministic fuzzy automaton, the Brzozowski type algorithm can produce a finite~one.

In addition to the determinization, practical applications of automata often require state reduction, a procedure of converting a given automaton into an equivalent automaton with a smaller~number~of~states. As the state minimization problem for fuzzy finite automata, as well as for nondeterministic~ones,~is~com\-pu\-tation\-ally hard ({\small PSPACE}-complete \cite{JR.93,LQ.13,Yu.97}), it is not required that this~equivalent automaton is~minimal, but it is necessary that it is effectively computable.~From different aspects, the state reduction~for~fuzzy~automata was studied in \cite{BG.02,CM.04,LL.07,MMS.99,P.04b,PZ.08,P.06,WQ.10,XQLF.07},~as well as in the books \cite{MM.02,PK.04}.~All algorithms~provided~there were motivated by the basic idea used in the minimization of ordinary de\-ter\-ministic automata, the idea of detecting and merging indistinguishable states, which boils down to com\-putation of certain~crisp equiva\-lences on the set of states.~A new approach to the state reduction~was initiated in \cite{CSIP.10,SCI.14}.~First, it was shown that better reductions of fuzzy finite auto\-mata can be achieved~if fuzzy equi\-valences are used instead of ordinary equivalences, and even better if fuzzy quasi-orders are used.~In addition, it was shown that the state reduction problem for fuzzy finite automata can be reduced to the problem of finding fuzzy quasi-orders that are solutions to a particular system of fuzzy relation equa\-tions, called the general system. As the general system is difficult to solve, the problem was further reduced to the search for instances of the general system and their solutions which ensure the best~possible reductions and can be efficiently computed.~Two such instances, whose solutions were called right~and~left invariant, have the greatest solutions that can be computed in polynomial time, and two others, whose solutions were called weakly right and left invariant, have the greatest solutions~that~ensure~better~reduc\-tions, but their computation requires exponential time.

The main aim of this paper is to combine determinization and state reduction methods into two-in-one algorithms that simultaneously perform determinization and state reduction.~These algorithms~perform better than all previous determinization algorithms for fuzzy finite automata, developed in \cite{Bel.02b,ICB.08,JIC.11,LP.05}, in the sense that they produce smaller automata, while require the same computation time.~The~only excep\-tion is the Brzozowski type~determinization algorithm developed recently in \cite{JC.13}, which produces~a minimal crisp-deterministic~fuzzy~automaton,~but we will see that the algorithms created here can be~used within the Brzozowski type algorithm and improve its performance.

Our main results are the following.~For any~fuzzy~finite~automaton $\cal A$ and a reflexive weakly right~invar\-iant fuzzy relation $\varphi$ on $\cal A$, we construct a crisp-determin\-is\-tic fuzzy automaton ${\cal A}_{\varphi}$ and prove that it is equivalent to $\cal A$.~If $\varphi $ is a weakly right invariant fuzzy quasi-order, we show that the same automaton~${\cal A}_{\varphi}$ would be produced if we first perform the state reduction of $\cal A$ by means of $\varphi $ and then construct the~Nerode automaton of this reduced
auto\-maton.~Furthermore, we show that automata ${\cal A}_{\varphi}$ determined by right invariant fuzzy quasi-orders are smaller than the Nerode automaton of $\cal A$, and that larger~right invariant fuzzy quasi-orders determine smaller crisp-deterministic fuzzy automata.~For~a~fuzzy~finite~auto\-maton $\cal A$ and a reflexive weakly right invariant fuzzy relation $\varphi$ on $\cal A$, we also introduce the concept of the children automaton of ${\cal A}_\varphi $ and prove that it is equivalent~to~$\cal A$.~In addition, if $\varphi $ is a right invariant fuzzy quasi-order, we prove that the children automaton of ${\cal A}_\varphi $ is smaller than ${\cal A}_\varphi $, and that~larger~right invariant fuzzy quasi-orders determine smaller children automata. We also show that weakly left invariant fuzzy quasi-orders  play a completely different role in the determinization.~Namely, if they are used to~reduce~the number of states prior the construction of the Nerode automaton, they will be unsuccessful because the Nerode automaton of the reduced fuzzy automaton would be the same as the Nerode automaton of~the\break original one.~However, we show that weakly left invariant fuzzy quasi-orders may be successful~in~com\-bi\-na\-tion
with the construction of the reverse Nerode automaton and that a two-in-one
algorithm can~be~pro\-vided which can improve performance of the Brzozowski
type algorithm for fuzzy finite automata.

The structure of the paper is as follows.~In Section 2 we recall basic notions and notation concerning fuzzy sets and relations, fuzzy automata and languages and crisp-deterministic fuzzy automata, we recall the concepts of the Nerode automaton and the reverse Nerode automaton, as well as the concepts~of~right and left invariant and weakly right and left invariant fuzzy relations.~Our main theoretical results are presented in Section 3, and in Section 4 we provide algorithms, perform the analysis of their computation time, and give characteristic computational examples.

\section{Preliminaries}
\subsection{Fuzzy sets and relations}

In this paper we use complete residuated lattices as
structures of membership values.~A {\it residuated~lattice\/} is
an algebra ${\cal L}=(L,\land ,\lor , \otimes ,\to , 0, 1)$ such
that
\begin{itemize}
\parskip=0pt\itemindent=1.4mm
\item[{\rm (L1)}] $(L,\land ,\lor , 0, 1)$ is a lattice with the
least element $0$ and the greatest element~$1$, \item[{\rm (L2)}]
$(L,\otimes ,1)$ is a commutative monoid with the unit $1$,
\item[{\rm (L3)}] $\otimes $ and $\to $ form an {\it adjoint
pair\/}, i.e., they satisfy the {\it adjunction property\/}: for
all $x,y,z\in L$,
\begin{equation}\label{eq:adj}
x\otimes y \le z \ \iff \ x\le y\to z .
\end{equation}
\end{itemize}
If, additionally, $(L,\land ,\lor , 0, 1)$ is a complete lattice, then ${\cal L}$ is called a {\it
complete residuated lattice\/}.

The algebra $(L,\lor ,\otimes, 0, 1)$ is a semiring, and it is denoted by ${\cal L}^*$ and called the \emph{semiring reduct\/} of $\cal L$.

The operations $\otimes $ (called {\it multiplication\/}) and $\to
$ (called {\it residuum\/}) are intended for modeling the
conjunc\-tion and implication of the corresponding logical calculus,
and supremum ($\bigvee $) and infimum ($\bigwedge $) are~intend\-ed
for modeling of the existential and general quantifier,
respectively. An operation $\lra $ defined by
\begin{equation}\label{eq:bires}
x\lra y = (x\to y) \land (y\to x),
\end{equation}
called {\it biresiduum\/} (or {\it biimplication\/}), is used for
modeling the equivalence of truth values.~For the basic properties of complete residuated lattices we refer to \cite{Bel.02a,BV.05}.

The most studied and applied structures of truth values, defined
on the real unit interval $[0,1]$ with\break $x\land y =\min
(x,y)$ and $x\lor y =\max (x,y)$, are the {\it {\L}ukasiewicz
structure\/} ($x\otimes y = \max(x+y-1,0)$, $x\to y=
\min(1-x+y,1)$), the {\it Goguen} ({\it product\/}) {\it
structure\/} ($x\otimes y = x\cdot y$, $x\to y= 1$ if $x\le y$
and~$=y/x$ otherwise) and the {\it G\"odel structure\/} ($x\otimes
y = \min(x,y)$, $x\to y= 1$ if $x\le y$ and $=y$
otherwise).
Another important set of truth values is the set
$\{a_0,a_1,\ldots,a_n\}$, $0=a_0<\dots <a_n=1$, with $a_k\otimes
a_l=a_{\max(k+l-n,0)}$ and $a_k\to a_l=a_{\min(n-k+l,n)}$. A
special case of the latter algebras is the two-element Boolean
algebra of classical logic with the support $\{0,1\}$.~The only
adjoint pair on the two-element Boolean algebra consists of the
classical conjunction and implication operations.~This structure
of truth values is called the {\it Boolean structure\/}.

A partially ordered set $P$ is said to satisfy the \emph{descending chain condition} (briefly \emph{DCC}) if every~descending sequence of elements of $P$ eventually terminates, i.e., if for every descending sequence $\{a_k\}_{k\in \Bbb N}$ of elements of $P$ there is $k\in \Bbb N$ such that $a_k=a_{k+l}$, for all $l\in \Bbb N$. In other words, $P$ satisfies DCC if there is no infinite descending chain in $P$.

In the sequel $\cal L$ will be a complete residuated
lattice.~A {\it fuzzy subset\/} of a set $A$ {\it over\/} ${\cal
L}$, or~simply a {\it fuzzy subset\/} of $A$, is any mapping from
$A$ into $L$.~Ordinary crisp subsets~of~$A$ are considered as
fuzzy subsets of $A$ taking membership values in the set
$\{0,1\}\subseteq L$.~Let $f$ and $g$ be two fuzzy subsets of
$A$.~The {\it equality\/} of $f$ and $g$ is defined as the usual
equality of mappings, i.e., $f=g$ if and only if $f(x)=g(x)$, for
every $x\in A$. The {\it inclusion\/} $f\le g$ is also defined
pointwise: $f\le g$ if and only if $f(x)\le g(x)$, for every $x\in
A$. Endowed with this partial order the set $L^A$ of all fuzzy
subsets of $A$ forms a complete residuated lattice, in which the
meet (intersection) $\bigwedge_{i\in I}f_i$ and the join (union)
$\bigvee_{i\in I}f_i$ of an arbitrary family $\{f_i\}_{i\in I}$ of
fuzzy subsets of $A$ are mappings from $A$ into $L$ defined by
\[
\left(\bigwedge_{i\in I}f_i\right)(x)=\bigwedge_{i\in I}f_i(x), \qquad \left(\bigvee_{i\in
I}f_i\right)(x)=\bigvee_{i\in I}f_i(x),
\]
and the {\it product\/} $f\otimes g$ is a fuzzy subset defined by
$f\otimes g (x)=f(x)\otimes g(x)$, for every $x\in A$.

 A {\it fuzzy relation\/} between sets $A$ and $B$ (in this order) is any mapping from
$A\times B$ to $L$, i.e. , any fuzzy subset of $A\times B$, and the equality, inclusion (ordering), joins and meets
of fuzzy relations are defined as for fuzzy sets. Set of all fuzzy relations between $A$ and $B$ will be denoted by $L^{A\times B}$. In particular,~a~fuzzy relation on a set $A$ is any function from
$A\times A$ to $L$, i.e., any fuzzy subset of $A\times A$.~The set of all fuzzy relations on $A$ will be denoted by $L^{A\times A}$. The \emph{reverse} or \emph{inverse} of a fuzzy relation $\alpha\in L^{A\times B}$ is a fuzzy~relation $\alpha^{-1}\in L^{B\times A}$ defined by $\alpha^{-1}(b, a) = \alpha(a, b)$, for all $a\in A$ and $b\in B$. A crisp relation is a fuzzy
relation which takes values only in the set $\{0, 1\}$, and if $\alpha$ is a crisp relation of $A$ to $B$, then expressions
$"\alpha(a, b) = 1"$ and $"(a, b)\in \alpha"$ will have the same meaning.

For non-empty sets $A$, $B$ and $C$, and fuzzy relations $\alpha\in L^{A\times B}$ and $\beta\in L^{B\times C}$,  their {\it
composition\/}~$\alpha\circ \beta\in L^{A\times C}$ is a fuzzy relation  defined by
\begin{equation}\label{eq:comp.rr}
(\alpha \circ \beta )(a,c)=\bigvee_{b\in B}\,\alpha(a,b)\otimes \beta(b,c),
\end{equation}
for all $a\in A$ and $c\in C$.
 For $f\in L^{A}$, $\alpha \in L^{A\times B}$ and $g\in L^{B}$, compositions $f\circ\alpha\in L^{B}$ and $\alpha\circ g\in L^{A}$ are fuzzy sets defined by
\begin{equation}\label{eq:comp.sr}
(f \circ \alpha)(b)=\bigvee_{a\in A}\,f(a)\otimes \alpha(a,b),\qquad
(\alpha \circ g)(a)=\bigvee_{b\in B}\,\alpha(a,b)\otimes g(b),
\end{equation}
for every $a\in A$ and $b\in B$.
Finally, the composition of two fuzzy sets $f,g\in L^{A}$ is an element $f\circ
g\in L$ (scalar) defined by
\begin{equation}\label{eq:comp.ss}
f \circ g =\bigvee_{a\in A}\,f(a)\otimes g(a) .
\end{equation}
When the underlying sets are finite, fuzzy relations can be interpreted as matrices and fuzzy sets as~vectors with entries in $L$, and then the composition
of fuzzy relations can be interpreted as the matrix product, compositions
of fuzzy sets and fuzzy relations as vector-matrix products, and the composition
of two fuzzy set as the scalar (dot) product.

It is easy to verify that the composition of fuzzy relations is associative,
i.e.,
\begin{equation}
(\alpha\circ\beta)\circ\gamma = \alpha\circ(\beta\circ\gamma),  \label{eq:comp.as}
\end{equation}
for all $\alpha\in L^{A\times B}$, $\beta\in L^{B\times C}$ and $\gamma\in
L^{C\times D}$, and
\begin{equation}
( f \circ\alpha) \circ\beta= f\circ (\alpha\circ\beta),\quad ( f \circ\alpha) \circ g= f\circ (\alpha\circ g),\quad (\alpha\circ\beta)\circ h = \alpha\circ(\beta\circ h) \label{eq:comp.as2}
\end{equation}
for all $\alpha\in L^{A\times B}$, $\beta\in L^{B\times C}$, $f\in L^{A}$, $g\in L^{B}$ and $h\in L^{C}$. Hence, all parentheses in (\ref{eq:comp.as})
and (\ref{eq:comp.as2}) can be~omitted.

Let $\alpha,\beta\in L^{A\times A}$ and $f,g\in L^A$.~The
{\emph{right residual\/}} of $\beta$ by $\alpha$ is a fuzzy relation $\alpha\backslash\beta\in L^{A\times A}$ and the~{\emph{left~residual\/}} of $\beta$ by $\alpha$ is a fuzzy relation $\beta/\alpha\in L^{A\times A}$ defined by
\begin{equation}
(\alpha\backslash\beta) (a,b)=\displaystyle \bigwedge_{c\in A}\alpha(c,a)\to \beta(c,b),\qquad \qquad (\beta/\alpha) (a,b)=\displaystyle \bigwedge_{c\in C}\alpha(b,c)\to \beta(a,c),
\end{equation}
whereas the {\emph{right residual\/}} of $g$ by $f$ is a fuzzy relation $f\backslash g\in L^{A\times A}$ and the {\emph{left residual\/}} of $g$ by $f$ is a fuzzy relation $g/f\in L^{A\times A}$ defined by
\begin{equation}
(f\backslash g) (a,b)=f(a)\to g(b) ,\qquad\qquad (g/f) (b,a)=f(a)\to g(b),
\end{equation}
for all $a,b\in A$.

A fuzzy relation $\varphi$ on a set  $A$ is said to be {\it reflexive\/},
if $\varphi (a,a)=1$, to be \emph{symmetric}, if $\varphi (a,b)=\varphi(b,a)$,~and~to be {\it transitive\/}, if $\varphi (a,b)\otimes \varphi (b,c)\le \varphi (a,c)$, for all $a,b,c\in A$. A reflexive and transitive
fuzzy relation is called a
{\it fuzzy quasi-order\/} (in some sources~{\it
fuzzy preorder\/}). A symmetric fuzzy quasi-order is a \emph{fuzzy~equivalence}. For a fuzzy quasi-order $\varphi $ on $A$ and an element $a\in A$,
the $\varphi $-{\it afterset\/} of $a$ is a fuzzy set $a\varphi\in L^A$ defined
by $a\varphi (b)=\varphi(a,b)$, and the $\varphi $-{\it foreset\/} of $a$ is a fuzzy set $\varphi a\in L^A$ defined
by $\varphi a(b)=\varphi(b,a)$, for every $b\in A$.
When $\varphi $ is interpreted as a matrix, then its aftersets are the rows,
and its foresets are the columns of this matrix. If $\varphi $ is a fuzzy equivalence, then the $\varphi $-{\it afterset\/} of $a$ coincides with the $\varphi $-{\it foreset\/} of $a$, and it is called a \emph{fuzzy equivalence class\/} of $a$.

\subsection{Fuzzy automata}

Throughout this paper, $\Bbb N$ denotes the set of natural numbers (without zero), $X$ is an (finite) alphabet, $X^+$ and $X^*$ denote, respectively, the free semigroup and the free monoid over $X$, $\varepsilon$ denotes the empty word in $X^*$, and if not noted otherwise, ${\cal L}$~is a complete~residua\-ted lattice.

A {\it fuzzy automaton over\/} $\cal L$ and $X$, or simply a {\it fuzzy automaton\/}, is a quadruple
${\cal A}=(A,\sigma^A,\delta^A,\tau^A )$, where~$A$ is a non-empty set, called  the
{\it set of states\/}, $\delta^A:A\times X\times A\to L$~is~a
fuzzy subset~of $A\times X\times A$, called the {\it fuzzy transition function\/}, and $\sigma^A: A\to L$ and $\tau^A : A\to L$ are fuzzy subsets of~$A$, called the {\it fuzzy set of initial states} and the {\it fuzzy set terminal states}, respectively.~We can
interpret $\delta^A (a,x,b)$ as the degree~to~which an~input letter $x\in X$~causes~a~transition from a state $a\in A$ into a
state $b\in A$, and we can interpret $\sigma^A(a)$ and $\tau^A(a)$ as the degrees to which $a$ is respectively an input state and a terminal state. For methodological reasons we allow the set of states $A$ to be infinite.~A~fuzzy auto\-maton whose set of states is finite is called a {\it fuzzy finite automaton\/}.~A fuzzy automaton over the Boolean structure is called a \emph{nondeterministic automaton\/} or a \emph{Boolean automaton\/}.

Define a family $\{\delta_x^A\}_{x\in X}$ of fuzzy relations on $A$ by $\delta^A_x(a,b) = \delta^A(a,x,b)$, for each $x\in X$, and all $a,b\in A$,~and extend this
family to the family $\{\delta_u^A\}_{u\in X^*}$ inductively, as follows:
$\delta^A_\varepsilon=\Delta_A$, where $\Delta_A$ is the crisp equality relation
on $A$, and
\begin{equation}\label{eq:du}
 \delta^A_{x_1x_2\dots x_n}=\delta_{x_1}\circ
\delta_{x_2}\circ\dots\circ \delta^A_{x_n}
\end{equation}
for all $n\in \Bbb N$, $x_1,x_2,\ldots ,x_n\in X$.~Members of this family
are called {\it fuzzy transiton relations\/} of $\cal A$.~Evidently, $\delta^A_{uv}= \delta^A_u\circ \delta^A_v$, for all $u,v\in X^*$.~In addition, define
families $\{\sigma^A_u\}_{u\in X^*}$ and $\{\tau^A_u\}_{u\in X^*}$ by
\begin{equation}\label{eq:su.tu}
\sigma^A_u=\sigma^A\circ \delta^A_u, \qquad \tau^A_u=\delta^A_u\circ \tau^A,
\end{equation}
for all $u\in X^*$.

We can visualize a fuzzy finite automaton ${\cal A}=(A,\sigma^A,\delta^A,\tau^A )$ representing it as a labelled directed~graph whose nodes are states of $\cal A$, an edge from a node $a$ to a node $b$ is labelled by pairs of the form $x/\delta^A_x(a,b)$, for any $x\in X$, and for any node $a$ we draw an arrow labelled by $\sigma^A(a)$ that enters this node, and an~arrow labelled by $\tau^A(a)$ coming out of this node.~For the sake of simplicity, we do not~draw~edges~whose~all~labels are of the form $x/0$, and incoming and outgoing arrows labelled by $0$. In particular, if $\cal A$ is a Boolean~auto\-maton, instead of any label of the form $x/1$ we write just $x$, initial states are marked by incoming arrows without any label, and terminal states are marked by double circles.

A {\it fuzzy language\/} in $X^*$ over ${\cal L}$, or
just a {\it fuzzy language\/}, is any fuzzy subset of~$X^*$, i.e., any function~from
$X^*$ into $L$.~The {\it fuzzy language recognized by a fuzzy automaton\/} ${\cal A}=(A,\sigma^A,\delta^A , \tau^A )$ is the fuzzy language $\lBrack{\cal A}\rBrack \in L^{X^*}$ defined by
\begin{equation}\label{eq:recog}
\lBrack{\cal A}\rBrack(u) = \bigvee_{a,b\in A} \sigma^{A} (a)\otimes \delta_u^A(a,b)\otimes \tau^A(b)=\sigma^A \circ \delta^A_u\circ \tau^A  ,
\end{equation}
for any $u\in X^*$.~In other words, the~membership degree of the word
$u$~to~the fuzzy language $\lBrack{\cal A}\rBrack$ is equal to the degree to which $\cal A$ recognizes or accepts
the word $u$.~Fuzzy automata $\cal A$~and~$\cal B$ are called   {\it language equivalent\/}, or just {\it equivalent\/}, if $\lBrack{\cal A}\rBrack=\lBrack{\cal B}\rBrack$.

~
Let ${\cal A}=(A,\sigma^A,  \delta^A,\tau^A)$ be a fuzzy automaton over $\cal L$ and $X$ and let $\varphi $ be a
fuzzy quasi-order~on~$A$.~The fuzzy automa\-ton ${\cal A}/\varphi =(A/\varphi ,\sigma^{A/\varphi },\delta^{A/\varphi },\tau^{A/\varphi })$ where $A/\varphi =\{\,a\varphi\mid a\in A\}$ is the set of all $\varphi $-aftersets, the fuzzy transition function
$\delta^{A/\varphi }:(A/\varphi) \times X\times (A/\varphi) \to L$ is given by
\begin{equation}\label{eq:aft.aut}
\delta^{A/\varphi }(a\varphi ,x,b\varphi )=\bigvee_{a',b'\in
A}\varphi (a,a')\otimes\delta_{x}^A(a',b')\otimes \varphi (b',b)=(\varphi \circ\delta_x^A\circ \varphi )(a,b),
\end{equation}
for all $a,b\in A$ and $x\in X$, and
the fuzzy set $\sigma^{A/\varphi }\in L^{A/\varphi }$ of initial~states~and the fuzzy
set $\tau^{A/\varphi }\in L^{A/\varphi }$~of~ter\-mi\-nal states are defined by
\begin{gather}
\sigma^{A/\varphi }(a\varphi ) = \bigvee_{a'\in A}\sigma^{A} (a')\otimes \varphi (a',a) = (\sigma^A\circ
\varphi )(a)  ,\quad a\in A \label{eq:sE} \\
\tau^{A/\varphi }(a\varphi ) = \bigvee_{a'\in A}\varphi (a,a')\otimes \tau^A(a') = (\varphi \circ\tau^A)(a)  , \quad a\in A\label{eq:tE}
\end{gather}
for all $a\in A$, is called the
{\it afterset fuzzy automaton\/} of $\cal A$ with respect to $\varphi $.~The {\it  foreset fuzzy automaton\/} of $\cal A$ with respect to~$\varphi $ is defined dually, but since it is isomorphic to the afterset fuzzy automaton,~we~will~work only with afterset fuzzy automata.

The cardinality of a fuzzy automaton ${\cal
A}=(A,\sigma^A,\delta^{A},\tau^A)$, in notation $|{\cal A} |$, is defined as the cardinality of its
set of states $A$. A fuzzy automaton ${\cal A}$ is called {\it minimal fuzzy automaton} of a fuzzy language $f \in L^{X^*}$ if $\lBrack {\cal A}\rBrack=f $
and $|{\cal A} |\leqslant|{\cal A'}|$, for every fuzzy automaton ${\cal A'}$ such that
$\lBrack {\cal A}'\rBrack=f $. A minimal fuzzy automaton recognizing a given
fuzzy language $f$ is not necessarily unique up to an isomorphism.~This is also true for
nondeterministic automata.

Let ${\cal A}=(A,\delta^A,\sigma^A,\tau^A)$ be a fuzzy automaton over $\cal L$ and $X$.~The {\it reverse fuzzy
automaton\/} of ${\cal A}$ is a fuzzy auto\-maton $\overline{\cal
A}=(A,\bar{\delta}^A,\bar{\sigma}^A,\bar{\tau}^A)$, where $\bar{\sigma}^A=\tau^A$, $\bar{\tau}^A=\sigma^A$, and $\bar{\delta}^A:
A\times X\times A \to L$ is defined by:
\[
\bar{\delta}^A(a,x,b)=\delta^A(b,x,a),
\]
for all $a,b\in A$ and $x\in X$.~Roughly speaking, the reverse fuzzy automaton $\overline{\cal A}$ is obtained from ${\cal A}$~by~exchang\-ing fuzzy sets of initial and final states and ``reversing'' all the transitions.~Due to the fact that the multiplication $\otimes$ is commutative, we have that $\bar{\delta}_{u}(a,b)=\delta_{\bar{u}}(b,a)$, for all $a,b\in A$ and $u\in X^*$.

The {\it reverse fuzzy language\/} of a fuzzy language $f\in L^{X^*}$ is a fuzzy language $\overline{f}\in L^{X^*}$ defined by $\overline{f}(u)=f(\bar{u})$, for each $u\in X^*$.~As $\overline{(\bar{u})} = u$ for all $u\in X^*$, we have that $\overline{(\overline{f})}=f$, for any fuzzy language $f$.~It is easy~to~see~that the reverse fuzzy automaton $\overline{\cal A}$ recognizes the reverse fuzzy language $\overline{\lBrack{\cal A}\rBrack}$ of the fuzzy language ${\lBrack{\cal A}\rBrack}$ recognized by $\cal A$, i.e., $\lBrack\overline{{\cal A}}\rBrack=\overline{\lBrack{\cal A}\rBrack}$.

For more information on fuzzy automata over complete residuated lattices we refer to \cite{CSIP.10,SCI.14,IC.10,ICB.08,ICB.09,Qiu.01,Qiu.06,WQ.10,XQL.09,ICBP.10}

\subsection{Crisp-deterministic fuzzy automata}\label{sec:cdfa}

Let ${\cal A}=(A,\sigma^A,\delta^A,\tau^A)$ be a fuzzy automaton over  $X$ and  ${\cal L}$.~The fuzzy transition function $\delta^A $ is called~{\it crisp-deter\-ministic\/} if for every $x\in X$ and every $a\in A$ there exists
$a'\in A$ such that $\delta^A_x(a,a')=1$, and $\delta^A_x(a,b)=0$, for all $b\in A\setminus
\{a'\}$. The fuzzy set of initial states $\sigma^A $ is called {\it crisp-deterministic\/}
if there exists $a_0\in A$ such that $\sigma^A(a_0)=1$, and $\sigma^A(a)=0$, for every
$a\in A\setminus \{a_0\}$.~If both $\sigma^A $ and $\delta^A $ are crisp-deterministic,
then $\cal A$ is called a {\it crisp-deterministic fuzzy automaton\/} (for short: {\it
cdfa\/}), and if it is finite, then it is called a {\it crisp-deterministic fuzzy finite automaton\/} (for short: {\it
cdffa\/}).

A crisp-deterministic fuzzy automaton can also be defined as a quadruple ${\cal A}=(A,\delta^A,a_0,\tau^A )$, where~$A$~is
a non-empty {\it set of states\/}, $\delta^A :A\times X\to A$ is a~{\it transition function\/},
$a_0\in A$ is an {\it initial state\/} and $\tau^A \in L^A$ is a {\it fuzzy set of terminal states\/}.~The
transition~function $\delta^A $ can be extended to a function $\delta^A_*:A\times X^*\to A$ in the following way:
$\delta^A_*(a,\varepsilon)=a$, for every $a\in A$, and $\delta^A_*(a,ux)=\delta^A(\delta^A_*(a,u),x)$, for
all $a\in A$, $u\in X^*$~and~$x\in X$. A~state $a\in A$ is called {\it accessible\/} if there exists $u\in X^*$ such that
$\delta^*(a_0,u)=a$.~If every state of $\cal A$ is accessible, then $\cal A$ is called an {\it
accessible crisp-deterministic fuzzy automaton\/}.

The initial state and transitions of a crisp-deterministic fuzzy automaton are graphically represented~as in the case of Boolean automata, and the fuzzy set of terminal states is represented as in the case of fuzzy finite automata.

Let ${\cal A}=(A,\delta^A,a_0,\tau^A)$ and ${\cal B}=(B,\delta^B,b_0,\tau^B)$ be crisp-deterministic fuzzy automata. A function $\phi :A\to B$ is called a {\it homomorphism\/} of $\cal A$ into ${\cal B}$ if $\phi (a_0)=b_0$, $\phi(\delta^A(a,x))=\delta^B(\phi(a),x)$ and $\tau^A(a)=\tau^B(\phi(a))$, for all $a\in A$ and $x\in X$. A bijective~homo\-mor\-phism is called an {\it isomorphism\/}.~If there is a surjective homomorphism of $\cal A$ onto $\cal B$, then $\cal B$ is said to be a {\it homomorphic image\/} of $\cal A$, and if there is an isomorphism of
$\cal A$ onto $\cal B$, then we say that $\cal A$ and $\cal B$ are \emph{isomorphic crisp-deterministic fuzzy automata\/} and we write ${\cal A}\cong {\cal B}$.

The~{\it language\/} of $\cal A$ is the fuzzy language $\lBrack
{\cal A}\rBrack : X^*\to L$ defined by
\begin{equation}\label{eq:cd-beh}
\lBrack {\cal A}\rBrack(u)=\tau (\delta^*(a_0,u)) \ .
\end{equation}
for every $u\in X^*$.~Obviously, the image of $\lBrack {\cal A}\rBrack$ is contained in the image
of $\tau $ which is finite if the set of states $A$ is finite.~A fuzzy language $f : X^*\to L$ is called {\it cdffa-recognizable\/} if
there is a crisp-deterministic fuzzy finite~auto\-maton $\cal A$ over $X$ and ${\cal L}$ such that
$\lBrack {\cal A}\rBrack=f $.~Then we say that $\cal A$ {\it recognizes\/} $f $.

A crisp-deterministic fuzzy automaton ${\cal A}$ is called a {\it minimal crisp-deterministic fuzzy automaton} of a~fuzzy language $f$ if $\lBrack {\cal A}\rBrack=f$ and $|{\cal A} |<|{\cal A'}|$, for any crisp-deterministic fuzzy automaton ${\cal A'}$ such that $\lBrack {\cal A}'\rBrack=f$.

Next, the \textit{Nerode automaton} of a fuzzy automaton ${\cal A}=(A,\sigma,\delta,\tau)$ is a crisp-deterministic fuzzy automaton ${\cal A}_{N}=(A_{N},\sigma^{A}_{\varepsilon},\delta_{N},\tau_{N})$,
where $A_{N}=\{\sigma^{A}_{u}\mid u\in X^{*}\}$, and $\delta_{N}:A_{N}\times X\longrightarrow A_{N}$ and $\tau_{N}:A_{N}\to L$ are defined~by
\[
\delta_{N}(\sigma^{A}_{u},x)=\sigma^{A}_{ux},\qquad\qquad \tau_{N}(\sigma^{A}_{u})=\sigma^{A}_{u}\circ\tau^{A},
\]
for every $u\in X^*$ and $x\in X$.~The Nerode automaton was first constructed in \cite{ICB.08}, where it was shown~that it is equivalent to the starting fuzzy automaton $\cal A$. The name "Nerode automaton" was introduced in \cite{ICBP.10}.

The \textit{reverse Nerode automaton} of a ${\cal A}$ is the Nerode automaton
of the reverse fuzzy automaton of $\cal A$, i.e., the crisp-deterministic fuzzy automaton ${\cal A}_{\overline N}=(A_{\overline N},\tau^A_\varepsilon,\delta_{\overline N},\tau_{\overline N})$,~where $A_{\overline N}=\{\tau^A_u\mid u\in X^*\}$, and $\delta_{\overline N}:A_{\overline N}\times X\to A_{\overline N}$ and $\tau_{\overline N}:A_{\overline N}\times L$ are defined by
\[
\delta_{\overline N}(\tau_u^A,x)=\tau_{xu}^A,\qquad\qquad \tau_{\overline N}(\tau_u^A)=\sigma^A\circ \tau_u^A,
\]
for all $u\in X^*$ and $x\in X$.~

\subsection{Right and left invariant fuzzy relations}

In the rest of the paper, a fuzzy relation on a fuzzy automaton will mean a fuzzy relation on its set~of~states.

Let ${\cal A}=(A,\sigma^{A},\delta^{A},\tau^{A})$ be a fuzzy automaton.~A fuzzy relation $\varphi$ on $A$ is a \emph{right invariant} if
\begin{align}
&\varphi\circ\delta^{A}_{x}\le \delta^{A}_{x}\circ\varphi,\qquad \text{for
each}\ x\in X,\label{eq:ri.d} \\
&\varphi\circ\tau^{A}\le\tau^{A},\label{eq:ri.t}
\end{align}
and it is called \emph{weakly right invariant} if
\begin{equation}
\varphi\circ\tau_u^{A}\le\tau_u^{A},\qquad \text{for each}\ u\in X^*.\label{eq:wri}
\end{equation}
Similarly we define the dual concepts.~A fuzzy relation $\varphi $ on $A$ is called \emph{left invariant} if
\begin{align}
&\delta^{A}_{x}\circ\varphi\le\varphi\circ\delta^{A}_{x} ,\qquad \text{for
each}\ x\in X,\label{eq:li.d} \\
&\sigma^A\circ \varphi\le\sigma^{A},\label{eq:li.s}
\end{align}
and it is called \emph{weakly left invariant} if
\begin{equation}
\sigma_u^{A}\circ\varphi\le\sigma_u^{A},\qquad \text{for each}\ u\in X^*.\label{eq:wli}
\end{equation}
It is easy to verify that every right invariant fuzzy relation is weakly
right invariant, and every left~invari\-ant fuzzy relation is weakly left invariant.~Note
that if $\varphi $ is reflexive, then $\varphi $ satisfies
(\ref{eq:wri}) if and only if it satisfies
\begin{equation}
\varphi\circ\tau_u^{A}=\tau_u^{A},\qquad \text{for each}\ u\in X^*,\label{eq:wri.eq}
\end{equation}
and $\varphi $ satisfies (\ref{eq:wli}) if and only if it satisfies
\begin{equation}
\sigma_u^{A}\circ\varphi=\sigma_u^{A},\qquad \text{for each}\ u\in X^*,\label{eq:wli.eq}
\end{equation}
and consequently, $\varphi $ satisfies (\ref{eq:ri.t}) if and only if it satisfies $\varphi\circ\tau^{A}=\tau^{A}$, and it satisfies (\ref{eq:li.s}) if and only if it satisfies $\sigma^{A}\circ\varphi=\sigma^{A}$.~Moreover, if $\varphi $ is a fuzzy quasi-order, then $\varphi $ satisfies
(\ref{eq:ri.d}) if and only if it satisfies
\begin{equation}
\varphi\circ\delta^{A}_{x}\circ\varphi= \delta^{A}_{x}\circ\varphi,\qquad \text{for each}\ x\in X,\label{eq:ri.d.eq}
\end{equation}
and $\varphi $ satisfies (\ref{eq:li.d}) if and only if it satisfies
\begin{equation}
\varphi\circ\delta^{A}_{x}\circ\varphi= \varphi\circ\delta^{A}_{x},\qquad \text{for each}\ x\in X.\label{eq:li.d.eq}
\end{equation}

Right and left invariant fuzzy quasi-orders and fuzzy equivalences were introduced in \cite{CSIP.10,SCI.14}, where they were used in the state reduction of fuzzy automata.~They are closely related to forward and backward simulations and bisimulations between fuzzy automata, which were studied in \cite{CIDB.12,CIJD.12}. Namely, a fuzzy quasi-order $\varphi $ on a fuzzy automaton $\cal A$ is right invariant if its reverse $\varphi^{-1}$ is a forward simulation of $\cal A$ into itself, and $\varphi $ is left invariant if $\varphi $ is a backward simulation of $\cal A$ into itself.

Weakly right and left invariant fuzzy quasi-orders were introduced in \cite{SCI.14}, and they were also used in the state reduction of fuzzy automata. They provide smaller automata than right invariant fuzzy quasi-orders, but are more difficult to compute. Weakly right and left invariant fuzzy quasi-orders and fuzzy equivalences are closely related to weak forward and backward simulations and bisimulations, which were studied in \cite{J.13}.

Algorithms for computing the greatest right and left invariant fuzzy quasi-orders on a fuzzy finite~auto\-maton, as well as algorithms for computing the greatest weakly right and left invariant ones, were~provided in \cite{SCI.14}. They are presented here in Section \ref{sec:alg}, together with an analysis of their computational~time.

\section{Theoretical results}

Let ${\cal A}=(A,\sigma^{A},\delta^{A},\tau^{A})$ be a fuzzy automaton and $\varphi$ a fuzzy relation on $A$.~For each $u\in X^{*}$ we define
a fuzzy set $\varphi_{u}: A\to L$ inductively, as follows: for the empty
word $\varepsilon$ and all $u\in X^*$ and $x\in X$ we set
\begin{equation}
  \varphi_{\varepsilon} = \sigma^{A}\circ \varphi,\hspace{0.2 in}\varphi_{ux} = \varphi_{u}\circ\delta^{A}_{x}\circ \varphi \label{fiu}
\end{equation}
Clearly, if $u=x_1\dots x_n$, where $x_1,\dots, x_n\in X$, then
\begin{equation}\label{fiu.d}
\varphi_{u}=\sigma^{A}\circ \varphi\circ\delta^{A}_{x_{1}}\circ \varphi\circ...\circ\delta^{A}_{x_{n}}\circ \varphi.
\end{equation}
Now, set $A_\varphi=\{\varphi_{u}\mid u\in X^{*}\}$, and define
$\delta_{\varphi}:A_{\varphi}\times X\to A_{\varphi}$ and $\tau_{\varphi}:A_{\varphi}\to L$ as follows:
\begin{equation}
  \delta_{\varphi}(\varphi_{u},x) = \varphi_{ux},\hspace{0.2 in}\tau_{\varphi}(\varphi_{u}) = \varphi_{u}\circ\tau^{A},\label{fi.delta.tau}
\end{equation}
for all $u\in X^*$ and $x\in X$.~If $\varphi_{u}=\varphi_{v}$, for some $u,v\in X^{*}$, then for each $x\in X$ we have that
\[
 \delta_{\varphi}(\varphi_{u},x)=\varphi_{ux}=\varphi_{u}\circ\delta_{x}\circ\varphi=\varphi_{v}\circ\delta_{x}\circ\varphi=\varphi_{vx}= \delta_{\varphi}(\varphi_{v},x),
\]
and hence, $\delta_{\varphi}$ is a well-defined function. It is clear
that $\tau_{\varphi}$ is also a well-defined function, and consequently,    ${\cal A}_{\varphi}=(A_{\varphi},\varphi_{\varepsilon},\delta_{\varphi},\tau_{\varphi})$ is a well-defined crisp-deterministic fuzzy automaton.

The main question that arises here is how to choose a fuzzy relation $\varphi $ so that the automaton ${\cal A}_\varphi $~is~equivalent to the original
automaton $\cal A$.~The following theorem gives an answer to this question.

\begin{theorem}\label{th.cdffa}\it
Let ${\cal A}=(A,\sigma^A,\delta^A,\tau^A)$ be a fuzzy automaton and  $\varphi$    a reflexive weakly right invariant fuzzy relation on $\cal A$.~Then ${\cal A}_{\varphi}=(A_{\varphi},\varphi_{\varepsilon},\delta_{\varphi},\tau_{\varphi})$ is an accessible crisp-deterministic fuzzy automaton equivalent to ${\cal A}$.
\end{theorem}

\begin{proof}
According to (\ref{eq:wri.eq}), by induction we easily prove that
\begin{equation}\label{eq:tau.u}
\tau_{x_1x_2\dots x_n}^A=\varphi\circ \delta_{x_1}^A\circ \varphi\circ \dots\circ
\delta_{x_n}^A\circ \varphi\circ \tau^A .
\end{equation}
for each $n\in \Bbb N$ and all $x_1,\dots, x_n\in X$.

Now, according to (\ref{eq:tau.u}), for each $u=x_1\dots x_n$, where $x_1,\ldots
,x_n\in X$, we have that
\begin{align*}
\lBrack{\cal A}_{\varphi}\rBrack(u) &=\tau_\varphi(\delta_\varphi(\varphi_\varepsilon,u))=\tau_\varphi(\varphi_u)=\varphi_u\circ
\tau^A=(\sigma^{A}\circ \varphi\circ\delta^{A}_{x_{1}}\circ \varphi\circ...\circ\delta^{A}_{x_{n}}\circ \varphi)\circ \tau^A=\\
&=\sigma^{A}\circ ( \varphi\circ\delta^{A}_{x_{1}}\circ \varphi\circ...\circ\delta^{A}_{x_{n}}\circ \varphi\circ \tau^A)=\sigma^A\circ \tau_u^A=\sigma^A\circ \delta_u^A\circ \tau^A=\lBrack{\cal A}\rBrack(u),
\end{align*}
and besides,
\[
\lBrack{\cal A}_{\varphi}\rBrack(\varepsilon)=\tau_\varphi(\varphi_\varepsilon)=\varphi_\varepsilon\circ
\tau^A=\sigma^{A}\circ \varphi\circ \tau^A=\sigma^{A}\circ  \tau^A=\lBrack{\cal A}\rBrack(\varepsilon),
\]
which means that $\lBrack{\cal A}_{\varphi}\rBrack=\lBrack{\cal A}\rBrack$.~Therefore,
${\cal A}_\varphi $ is equivalent to $\cal A$.
\end{proof}

In addition, the following is true.

\begin{theorem}\label{th:det.red}\it
Let ${\cal A}=(A,\sigma^A,\delta^A,\tau^A)$ be a fuzzy automaton and $\varphi $ a weakly right invariant fuzzy quasi-order on~$\cal A$. Then the automaton ${\cal A}_\varphi $ is isomorphic to the Nerode automaton of the afterset fuzzy automaton ${\cal A}/\varphi $.
\end{theorem}

\begin{proof}
For the sake of simplicity, set $B=A/\varphi$ and ${\cal B}={\cal A}/\varphi$, i.e., let ${\cal B}=(B,\sigma^B,\delta^B,\tau^B)$ be the afterset~fuzzy~automaton of $\cal A$ corresponding to $\varphi $.~Consider the Nerode automaton ${\cal B}_N =(B_N,\sigma^B_\varepsilon,\delta_N,\tau_N)$ of $\cal B$.

First, by induction on the length of a word, we will prove that for any $u\in X^*$ the following is true:
\begin{equation}\label{eq:ner.aft}
\sigma_u^{B}(a\varphi)=\varphi_u(a), \ \ \text{for every}\ a\in A.
\end{equation}
For any $a\in A$ we have that $\sigma_\varepsilon^{B}(a\varphi )=\sigma^{B}(a\varphi )=(\sigma^A \circ \varphi)(a)=\varphi_\varepsilon(a)$, and thus, (\ref{eq:ner.aft})
holds when $u$ is the empty word.~Next, suppose that (\ref{eq:ner.aft}) holds for some word $u\in X^*$.~By (\ref{fiu.d}) and idempotency of $\varphi $ it follows that $\varphi_u\circ \varphi=\varphi_u$, so for each $x\in X$ and each $a\in A$ we have that
\begin{align*}
\sigma_{ux}^{B}(a\varphi)&=(\sigma_{u}^{B}\circ \delta_x^{B})(a\varphi)=\bigvee_{b\in A}\sigma_{u}^{B}(b\varphi)\otimes \delta_x^{B}(b\varphi,a\varphi)= \bigvee_{b\in A}\varphi_{u}(b)\otimes (\varphi\circ \delta_x^A\circ \varphi)(b,a)= \\
&= (\varphi_u\circ \varphi \circ \delta_x^A\circ \varphi )(a)=(\varphi_u\circ \delta_x^A\circ \varphi )(a)=\varphi_{ux}(a).
\end{align*}
Therefore, we conclude that (\ref{eq:ner.aft}) holds for every $u\in X^*$.

Now, define a function $\xi :A_\varphi \to B_N$ by $\xi (\varphi_u)=\sigma_u^B$, for each $u\in X^*$. For arbitrary $u,v\in X^*$ we have that
\[
\varphi_u=\varphi_v\ \iff\ (\forall a\in A)\ \varphi_u(a)=\varphi_v(a) \ \iff\ (\forall a\in A)\ \sigma_u^B(a\varphi )=\sigma_v^B(a\varphi ) \ \iff\ \sigma_u^B=\sigma_v^B,
\]
so $\xi $ is a well-defined and injective function.~It is clear that $\xi $ is also surjective, and therefore, $\xi $ is a bijective function. Also, for all
for all $u\in X^*$ and~$x\in X$ we have that
\begin{align*}
&\delta_N(\xi(\varphi_u),x)=\delta_N(\sigma_u^B,x)=\sigma_{ux}^B=\xi(\varphi_{ux})=\xi (\delta_\varphi(\varphi_u,x)), \\
&\tau_N(\xi(\varphi_u))= \tau_N(\sigma_u^B)= \sigma_u^B\circ \tau^B=\lBrack{\cal B}\rBrack (u)=\lBrack{\cal A}\rBrack (u)=\lBrack{\cal A}_\varphi\rBrack (u)=\varphi_u\circ\tau^A=\tau_\varphi(\varphi_u),
\end{align*}
so $\xi $ is an isomorphism of the automaton ${\cal A}_\varphi $ onto the Nerode automaton ${\cal B}_N$ of ${\cal B}={\cal A}/\varphi$.
\end{proof}

\begin{remark}\rm
Note that in the previous theorem we need $\varphi $ to be weakly right invariant
only to prove that $\tau_N(\xi(\varphi_u))=\tau_\varphi(\varphi_u)$. Everything
else can be proved under the weaker assumpton that $\varphi $ is a fuzzy
quasi-order.
\end{remark}

In the case when working with right invariant fuzzy quasi-orders, it is possible to compare the size of the corresponding automata.~This follows from the following theorem.

\begin{theorem}\label{th:hom.im}\it
Let ${\cal A}=(A,\sigma^A,\delta^A,\tau^A)$ be a fuzzy automaton and let $\varphi $ and $\phi $ be right invariant fuzzy quasi-orders on $\cal A$ such that $\varphi \leqslant \phi $.~Then the automaton ${\cal A}_{\phi} $ is a homomorphic image of the automaton ${\cal A}_\varphi $.

Consequently, $|{\cal A}_\phi|\leqslant |{\cal A}_\varphi |$.
\end{theorem}

\begin{proof}
First we note that $\varphi \leqslant \phi $ is equivalent to $\varphi \circ \phi =\phi \circ \varphi =\phi $, because $\varphi $ and $\phi $ are fuzzy quasi-orders.

Define a function $\xi :A_\varphi \to A_{\phi} $ by $\xi (\varphi_u)=\phi_u$, for each $u\in X^*$.~First we prove that $\xi $ is well-defined.~Let $u,v\in X^*$ such that $\varphi_u=\varphi_v$.~According to (\ref{eq:ri.d.eq}) and (\ref{fiu.d}), by induction we easily prove that $\varphi_w=\sigma_w^A\circ \varphi $ and $\phi_w=\sigma_w^A\circ \phi $, for every $w\in X^*$, whence
\[
\phi_u=\sigma_u^A\circ \phi =\sigma_u^A\circ \varphi\circ \phi = \varphi_u\circ \phi = \varphi_v\circ \phi = \sigma_v^A\circ \varphi\circ \phi = \sigma_v^A\circ \phi =\phi_v .
\]
Therefore, $\xi $ is a well-defined function.~It is clear that $\xi $ is a surjective function.~Moreover, it is evident that $\delta_\phi (\xi(\varphi_u),x)=\xi(\delta_\varphi(\varphi_u,x))$, for all $u\in X^*$ and $x\in X$, and
\[
\tau_{\phi} (\xi(\varphi_u))=\tau_\phi (\phi_u)=\phi_u\circ \tau^A=\lBrack{\cal A}_{\phi}\rBrack (u)=\lBrack{\cal A}\rBrack (u)=\lBrack{\cal A}_{\varphi}\rBrack (u)=
\varphi_u\circ \tau^A=\tau_\varphi(\varphi_u),
\]
and hence, $\xi $ is a homomorphism of ${\cal A}_\varphi $ onto ${\cal A}_{\phi}$
and $|{\cal A}_\phi|\leqslant |{\cal A}_\varphi |$.
\end{proof}

Note that when $\varphi $ is a reflexive weakly left invariant fuzzy relation on a fuzzy automaton ${\cal A}$,~then ${\cal A}_\varphi $ is just the Nerode automaton of $\cal A$, and we do not get any new construction.~Besides, the following is true.

\begin{theorem}\label{th:det.red.wli}\it
Let ${\cal A}=(A,\sigma^A,\delta^A,\tau^A)$ be a fuzzy automaton and $\varphi $ a weakly left invariant fuzzy quasi-order on~$\cal A$.~Then the Nerode automaton of the afterset fuzzy automaton ${\cal A}/\varphi $ is isomorphic to the Nerode automaton~of~$\cal A$.
\end{theorem}

\begin{proof}For the sake of simplicity, set $B=A/\varphi$ and ${\cal B}={\cal A}/\varphi$, i.e., let ${\cal B}=(B,\sigma^B,\delta^B,\tau^B)$ be the afterset~fuzzy~automaton of $\cal A$ corresponding to $\varphi $.

First, by induction on the length of a word, we will prove that for any $u\in X^*$ the following is true:
\begin{equation}\label{eq:ner.wli}
\sigma_u^{B}(a\varphi)=\sigma_u^A(a), \ \ \text{for every}\ a\in A.
\end{equation}
For every $a\in A$ we have that $\sigma_\varepsilon^{B}(a\varphi )=(\sigma^A \circ\varphi)(a)= \sigma^A(a)= \sigma_\varepsilon^A(a)$, so (\ref{eq:ner.wli})
holds when $u$ is the empty~word. Next, suppose that (\ref{eq:ner.wli}) holds for some word $u\in X^*$.~According to (\ref{eq:ner.wli}) and our starting hypothesis that $\varphi $ is a weakly left invariant fuzzy quasi-order, for all $x\in X$ and $a\in A$ we obtain that
\begin{align*}
\sigma_{ux}^B(a\varphi )&=(\sigma_u^B\circ \delta_x^B)(a\varphi )=\bigvee_{b\in A}\sigma_u^B(b\varphi) \otimes \delta_x^B(b\varphi ,a\varphi )=\bigvee_{b\in A}\sigma_u^A(b) \otimes (\varphi\circ\delta_x^A\circ \varphi )(b,a)\\
&=(\sigma_u^A\circ \varphi\circ \delta_x^A\circ\varphi )(a)=(\sigma_u^A\circ \delta_x^A\circ\varphi )(a)=(\sigma_{ux}^A\circ \varphi)(a)= \sigma_{ux}^A(a),
\end{align*}
what completes the proof of (\ref{eq:ner.wli}).

Now, define a function $\xi :B_N\to A_N$ by $\xi (\sigma_u^B)=\sigma_u^A$, for each $u\in X^*$. For arbitrary $u,v\in X^*$ we have that
\[
\sigma_u^B=\sigma_v^B\ \iff\ (\forall a\in A)\ \sigma_u^B(a\varphi)=\sigma_v^B(a\varphi) \ \iff\ (\forall a\in A)\ \sigma_u^A(a)=\sigma_v^A(a) \ \iff\ \sigma_u^A=\sigma_v^a,
\]
which means that $\xi $ is a well-defined and injective function.~In addition, $\xi $ is surjective, and consequently, $\xi $ is a bijective function. It is easy to check that $\xi $ is a homomorphism, and therefore, $\xi $ is an isomorphism of the Nerode automaton of $\cal B$ onto the Nerode automaton of $\cal A$.
\end{proof}

As we said earlier, when $\varphi $ is a reflexive weakly left invariant fuzzy relation on a fuzzy automaton ${\cal A}$, then ${\cal A}_\varphi $ is just the Nerode automaton of $\cal A$, and we do not get any new construction.~However, we will show that weakly left invariant fuzzy relations work well with another construction, and can be very useful in the determinization of the reverse fuzzy automaton of $\cal A$.

Let ${\cal A}=(A,\sigma^{A},\delta^{A},\tau^{A})$ be a fuzzy automaton and $\psi$ a fuzzy relation on $A$.~For each $u\in X^{*}$ we define
a fuzzy set $\psi^{u}: A\to L$ inductively, as follows: for the empty
word $\varepsilon$ and all $u\in X^*$ and $x\in X$ we set
\begin{equation}
  \psi^{\varepsilon} = \psi \circ \tau^{A},\hspace{0.2 in}\psi^{xu} =\psi \circ\delta^{A}_{x}\circ\psi^{u} \label{psiu}
\end{equation}
Clearly, if $u=x_1\dots x_n$, where $x_1,\dots, x_n\in X$, then
\begin{equation}\label{psiu.d}
\psi^{u}= \psi\circ\delta^{A}_{x_1}\circ \psi\circ...\circ\delta^{A}_{x_{n}}\circ \psi\circ \tau^A.
\end{equation}
Now, set $A^\psi=\{\psi_{u}\mid u\in X^{*}\}$, and define
$\delta^{\psi}:A^{\psi}\times X\to A^{\psi}$ and $\tau^{\psi}:A^{\psi}\to L$ as follows:
\begin{equation}
  \delta^{\psi}(\psi^{u},x) = \psi^{xu},\hspace{0.2 in}\tau^{\psi}(\psi^{u}) =\sigma^A\circ \psi^{u},\label{psi.delta.tau}
\end{equation}
for all $u\in X^*$ and $x\in X$.~If $\psi_{u}=\psi_{v}$, for some $u,v\in X^{*}$, then for each $x\in X$ we have that
\[
 \delta^{\psi}(\psi^{u},x)=\psi^{xu}=\psi \circ\delta^{A}_{x}\circ\psi^{u}=\psi \circ\delta^{A}_{x}\circ\psi^{v}=\psi^{xv}= \delta^{\psi}(\psi^{v},x),
\]
and hence, $\delta^{\psi}$ is a well-defined function.~Clearly, $\tau^{\psi}$ is also a well-defined function, so ${\cal A}^{\psi}=(A^{\psi},\psi^{\varepsilon},\delta^{\psi},\tau^{\psi})$~is a well-defined crisp-deterministic fuzzy automaton.

Now we prove that the following is true.

\begin{theorem}\label{th.cdffa.wli}\it
Let ${\cal A}=(A,\sigma^A,\delta^A,\tau^A)$ be a fuzzy automaton and  $\psi$ a reflexive weakly left invariant fuzzy relation on $\cal A$.~Then ${\cal A}^{\psi}=(A^{\psi},\psi^{\varepsilon},\delta^{\psi},\tau^{\psi})$ is an accessible crisp-deterministic fuzzy automaton equivalent to the reverse
fuzzy automaton of ${\cal A}$.
\end{theorem}

\begin{proof}
Consider an arbitrary word $u=x_1\dots x_n$, where $x_1,\ldots ,x_n\in X$.~Using (\ref{psiu.d}) we obtain that
\[
\sigma^A\circ \psi \circ \delta_{x_n}^A\circ \psi\circ \dots \circ \delta_{x_1}^A\circ \psi = \sigma^A_{x_n\dots x_1},
\]
whence it follows that
\begin{align*}
\lBrack {\cal A}^\psi \rBrack (u)&= \tau^\psi (\delta^\psi(\psi^\varepsilon ,u))= \tau^\psi (\psi^{\bar u})=\sigma^A\circ \psi^{\bar u}=\sigma^A\circ
(\psi\circ\delta^{A}_{x_{n}}\circ \psi\circ...\circ\delta^{A}_{x_{1}}\circ \psi\circ \tau^A)= \\
&= (\sigma^A\circ \psi\circ\delta^{A}_{x_{n}}\circ \psi\circ...\circ\delta^{A}_{x_{1}}\circ \psi)\circ \tau^A=\sigma^A_{x_n\dots x_1}\circ \tau^A=\sigma_{\bar u}^A\circ \tau_A=\lBrack {\cal A}\rBrack (\bar u)=\lBrack \overline{\cal A}\rBrack (u).
\end{align*}
On the other hand,
\[
\lBrack {\cal A}^\psi \rBrack (\varepsilon)=\tau^\psi (\psi^\varepsilon)=\sigma^A\circ \psi^\varepsilon=\sigma^A\circ \psi \circ \tau^A=\sigma^A\circ \tau^A=
\lBrack {\cal A}\rBrack (\varepsilon)=\lBrack \overline{\cal A}\rBrack (\varepsilon).
\]
Therefore, $\lBrack {\cal A}^\psi \rBrack=\lBrack \overline{\cal A} \rBrack$, i.e., ${\cal A}^\psi $ is equivalent to $\overline{\cal A}$.
\end{proof}

The next two theorems can be proved similarly as Theorems \ref{th:det.red} and \ref{th:hom.im}, so their proofs will be omitted.

\begin{theorem}\label{th:det.red.wli}\it
Let ${\cal A}=(A,\sigma^A,\delta^A,\tau^A)$ be a fuzzy automaton and $\psi $ a weakly left invariant fuzzy quasi-order on~$\cal A$. Then the automaton ${\cal A}^\psi $ is isomorphic to the reverse Nerode automaton of the afterset fuzzy automaton~${\cal A}/\psi $.
\end{theorem}

\begin{theorem}\label{th:hom.im.li}\it
Let ${\cal A}=(A,\sigma^A,\delta^A,\tau^A)$ be a fuzzy automaton and let $\psi $ and $\psi' $ be left invariant fuzzy quasi-orders on $\cal A$ such that $\psi \leqslant \psi' $.~Then the automaton ${\cal A}^{\psi'} $ is a homomorphic image of the automaton ${\cal A}^\psi $.
\end{theorem}

Note that the reverse Nerode automaton plays a crucial role in Brzozowski type determinization of~a fuzzy automaton.~Namely, it has been proven in \cite{JC.13} that when we start from a fuzzy automaton $\cal A$,~two consecutive applications of the construction of a reverse Nerode automaton produce a minimal crisp-deter\-ministic fuzzy automaton which is equivalent to $\cal A$.~We will show here that the first of these two~constructions can be replaced by construction of the automaton ${\cal A}^\psi $, for some reflexive weakly left invariant fuzzy relation $\psi $ on $\cal A$.

\begin{theorem}\label{th:Brz.wli}\it
Let ${\cal A}=(A,\sigma^A,\delta^A,\tau^A)$ be a fuzzy automaton and  $\psi$ a reflexive weakly left invariant fuzzy relation on $\cal A$.~Then the reverse
Nerode automaton of ${\cal A}^{\psi}$ is a minimal crisp-deterministic fuzzy automaton equivalent~to~${\cal A}$.
\end{theorem}

\begin{proof} As we have proved in Theorem \ref{th.cdffa.wli}, ${\cal A}^\psi $ is an accessible crisp-deterministic fuzzy automaton equivalent to $\overline{\cal A}$.~According to Theorem 3.5 \cite{JC.13}, for any accessible crisp-deterministic fuzzy automaton $\cal B$, the reverse Nerode automaton of $\cal B$ is a minimal crisp-deterministic fuzzy automaton equivalent to $\overline{\cal B}$. Therefore, ${\cal A}^{\psi}$ is a minimal crisp-deterministic fuzzy automaton equivalent to the reverse of $\overline{\cal A}$, i.e., it is a minimal crisp-deterministic fuzzy automaton equivalent to ${\cal A}$.
\end{proof}

As the automaton ${\cal A}^\psi $ can be significantly smaller than the reverse Nerode automaton ${\cal A}_{\overline N}$ (in particular, if $\psi $ is the greatest left invariant fuzzy quasi-order on $\cal A$), replacing ${\cal A}_{\overline N}$ with ${\cal A}^\psi $ in the first step of Brzozowski type procedure we could mitigate a combinatorial blow up of the number of states that may happen in this step.~In the second step, such a problem does not exist because both constructions give the minimal crisp-deterministic fuzzy automaton equivalent to ${\cal A}$.

Let ${\cal A}=(A,\sigma^A,\delta^A,\tau^A)$ be a fuzzy automaton over an
alphabet $X=\{x_1,\ldots ,x_m\}$ and $\varphi
$ a fuzzy relation~on~$A$. For each $u\in X^*$ define an $(m+1)$-tuple
$\varphi_u^c$ by
\[
\varphi_u^c=(\varphi_{ux_1},\ldots ,\varphi_{ux_m},\varphi_u\circ
\tau^A)=(\varphi_u\circ \delta_{x_1}^A,\ldots , \varphi_u\circ \delta_{x_m}^A,\varphi_u\circ
\tau^A),
\]
set $A_\varphi^c=\{\,\varphi_u^c\mid u\in X^*\}$, and define $\delta_\varphi^c:A_\varphi^c\times
X\to A_\varphi^c$ and $\tau_\varphi^c:A_\varphi^c\to L$ as follows:
\begin{equation}\label{eq:A.fc}
\delta_\varphi^c(\varphi_u^c,x)=\varphi_{ux}^c, \qquad \tau_\varphi^c(\varphi_u^c)=\varphi_u\circ
\tau^A ,
\end{equation}
for all $u\in X^*$ and $x\in X$. We have the following:

\begin{theorem}\label{th:A.fc}\it
Let ${\cal A}=(A,\sigma^A,\delta^A,\tau^A)$ be a fuzzy automaton over an
alphabet $X=\{x_1,\ldots ,x_m\}$ and $\varphi
$ a reflexive weakly right invariant fuzzy relation~on~$A$.~Then ${\cal A}_\varphi^c=(A_\varphi^c,\varphi_\varepsilon^c,\delta_\varphi^c,\tau_\varphi^c)$
is an acessible crisp-deterministic fuzzy automaton equivalent to $\cal A$. \end{theorem}

\begin{proof}
Let $\varphi_u^c=\varphi_v^c$, for some $u,v\in X^*$.~This means that $\varphi_{ux_i}=\varphi_{vx_i}$, for any $i\in \{1,\ldots,m\}$, and $\varphi_u\circ \tau^A=\varphi_v\circ \tau^A$.~Now, for an arbitrary $x\in X$ we have that $\varphi_{ux}=\varphi_{vx}$, whence
\[
\varphi_{uxx_i}=\varphi_{ux}\circ \delta_{x_i}^A\circ \varphi = \varphi_{vx}\circ \delta_{x_i}^A\circ \varphi =\varphi_{vxx_i},
\]
for each $i\in \{1,\ldots,m\}$, and also, $\varphi_{ux}\circ \tau^A=\varphi_{vx}\circ \tau^A$.~Hence, $\varphi_{ux}^c=\varphi_{vx}^c$, so $\delta_\varphi^c(\varphi_u^c,x)=\varphi_{ux}^c=\varphi_{vx}^c=\delta_\varphi^c(\varphi_v^c,x)$, and this means that $\delta_\varphi^c$ is a well-defined function.~Clearly, $\tau_\varphi^c$ is also a well-defined function, and consequently, ${\cal A}_\varphi^c$ is an accessible crisp-deterministic automaton.

Next, for each $u\in X^*$ we have that
\[
\lBrack {\cal A}_\varphi^c\rBrack (u) = \tau_\varphi^c(\delta_\varphi^c(\varphi_\varepsilon^c,u)) = \tau_\varphi^c(\varphi_u^c) = \varphi_u\circ\tau^A=\lBrack {\cal A}_\varphi\rBrack (u)= \lBrack {\cal A}\rBrack (u),
\]
and therefore, $\lBrack {\cal A}_\varphi^c\rBrack$ is equivalent to $\cal A$.
\end{proof}

In the case when $\varphi $ is the crisp equality on $A$, we have that $\varphi_u=\sigma_u^A$, for each $u\in X^*$, and ${\cal A}_\varphi^c$~is~the~auto\-maton constructed in \cite{JIC.11}, where it was called the {\it reduced Nerode automaton\/} of $\cal A$.~Here we will use~a~differ\-ent terminology. Namely, the first $m$ elements in the $(m+1)$-tuple $\varphi_u^c$ are the children of the vertex $\varphi_u$~in~the transition tree of the automaton ${\cal A}_\varphi $ (see Algorithm \ref{alg:A.phi}), and the $m+1$st element of $\varphi_u^c$ is the termination degree of $\varphi_u$ in ${\cal A}_\varphi $.~For this reason, the automaton ${\cal A}_\varphi^c$ will be called the {\it children automaton\/}~of~${\cal A}_\varphi $.~In~this regard, the above mentioned automaton constructed in \cite{JIC.11} is the children automaton of the Nerode auto\-maton ${\cal A}_N$ of $\cal A$.~The children automaton of the Nerode automaton ${\cal A}_N$ is denoted by ${\cal A}_N^c=(A_N^c,\sigma_{\varepsilon}^c,\delta_N^c,\tau_N^c)$.

\begin{theorem}\label{th:Afc.Afc}\it
Let ${\cal A}=(A,\sigma^A,\delta^A,\tau^A)$ be a fuzzy automaton over an
alphabet $X=\{x_1,\ldots ,x_m\}$ and let $\varphi $ and $\phi $ be right invariant fuzzy quasi-orders~on~$A$ such that $\varphi\leqslant \phi $.~Then
the automaton ${\cal A}_\phi^c$ is a homomorphic image of the automaton ${\cal A}_\varphi^c$, and consequently, $|{\cal A}_\phi^c|\leqslant |{\cal A}_\varphi^c|$.
\end{theorem}

\begin{proof}
Define a function $\xi:A_\varphi^c\to A_\phi^c$ by $\xi(\varphi_u^c)=\phi_u^c$,
for each $u\in X^*$. Let $\varphi_u^c=\varphi_v^c$, for some~$u,v\in X^*$,~i.e., let $\varphi_{ux}=\varphi_{vx}$, for eny $x\in X$, and $\varphi_u\circ
\tau^A=\varphi_v\circ \tau^A$.~As in the proof of Theorem \ref{th:hom.im}
we obtain that $\varphi_w=\sigma_w^A\circ \varphi $ and $\phi_w=\sigma_w^A\circ \phi$, for every $w\in X^*$, and for any $x\in X$ we have that
\[
\phi_{ux}=\sigma_{ux}^A\circ \phi = \sigma_{ux}^A\circ \varphi\circ \phi =\varphi_{ux}\circ \phi = \varphi_{vx}\circ \phi = \sigma_{vx}^A\circ \varphi\circ \phi= \sigma_{vx}^A\circ  \phi=\phi_{vx},
\]
and also,
\[
\phi_u\circ \tau^A = \lBrack{\cal A}_\phi\rBrack (u)=\lBrack{\cal A}\rBrack (u)=\lBrack{\cal A}_\varphi\rBrack (u)=\varphi_u\circ \tau^A = \varphi_v\circ \tau^A =\lBrack{\cal A}_\varphi\rBrack (v)=\lBrack{\cal A}_\phi\rBrack (v)=\phi_v\circ \tau^A .
\]
Therefore, $\phi_u^c=\phi_v^c$, which means that $\xi $ is a well-defined
function, and clearly, $\xi $ is surjective.~Moreover,~for all $u\in X^*$
 and $x\in X$ we have that
 \[
\xi(\delta_\varphi^c(\varphi_u^c,x))=\xi(\varphi_{ux}^c)= \phi_{ux}^c=\delta_\phi^c(\phi_u^c,x)=\delta_\phi^c(\xi(\varphi_u^c),x),  \]
and, on the other hand,
\[
\tau_\phi^c(\phi_u^c)=\phi_u\circ \tau^A=\lBrack{\cal A}_\phi\rBrack (u)=\lBrack{\cal A}\rBrack (u)=\lBrack{\cal A}_\varphi\rBrack (u)=\varphi_u\circ \tau^A =\tau_\varphi^c(\varphi_u^c).
\]
Hence, $\xi $ is a homomorphism of ${\cal A}_\varphi$ onto ${\cal A}_\phi$.
\end{proof}

\begin{theorem}\label{th:Afc.afters}\it
Let ${\cal A}=(A,\sigma^A,\delta^A,\tau^A)$ be a fuzzy automaton over an
alphabet $X=\{x_1,\ldots ,x_m\}$ and let $\varphi $ be a weakly right invariant fuzzy quasi-order~on~$\cal A$.~Then
the  automaton ${\cal A}_\varphi^c$ is isomorphic to the automaton ${\cal B}_N^c$,~where ${\cal B}={\cal A}/\varphi $ is the afterset fuzzy automaton of $\cal A$ with respect to $\varphi $.
\end{theorem}

\begin{proof}
Define a function $\xi :A_\varphi^c\to B_N^c$ by $\xi (\varphi_u^c)=\sigma_u^{B,c}=(\sigma_{ux_1}^B,\ldots
,\sigma_{ux_m}^B,\sigma_u^B\circ \tau^B)$, for each $u\in X^*$.~As~in~the
proof of Theorem \ref{th:det.red} we obtain that $\sigma_u^B(a\varphi)=\varphi_u(a)$,
for all $u\in X^*$ and $a\in A$, and by this it follows that
\[
\varphi_{ux_i}=\varphi_{vx_i}\ \ \iff\ \ (\forall a\in A)\ \varphi_{ux_i}(a)=\varphi_{vx_i}(a)
\ \ \iff\ \ (\forall a\in A)\ \sigma_{ux_i}^B(a\varphi)=\sigma_{vx_i}^B(a\varphi)
\ \ \iff \ \ \sigma_{ux_i}^B=\sigma_{vx_i}^B ,
\]
and also,
\begin{align*}
\varphi_u\circ \tau^A=\varphi_v\circ \tau^A \ \ &\iff\ \ \lBrack {\cal A}_\varphi
\rBrack (u)=\lBrack {\cal A}_\varphi
\rBrack (v)\ \ \iff\ \ \lBrack {\cal A}
\rBrack (u)=\lBrack {\cal A} \rBrack (v)\\
&\iff\ \ \lBrack {\cal B}
\rBrack (u)=\lBrack {\cal B}
\rBrack (v)\ \ \iff\ \ \lBrack {\cal B}_N
\rBrack (u)=\lBrack {\cal B}_N
\rBrack (v)\ \ \iff\ \ \sigma_u^B\circ \tau^B=\sigma_v^B\circ \tau^B,
\end{align*}
for all $u,v\in X^*$ and $x_i\in X$, and therefore, $\varphi_u^c=\varphi_v^c$
if and only if $\sigma_u^{B,c}=\sigma_v^{B,c}$.~This means that $\xi $ is~a~well-defined
and injective function.~It is clear that $\xi $ is also surjective, and iot
can be easily verified that it is a homomorphism. Hence, $\xi $ is an isomorphism
of  ${\cal A}_\varphi^c$ onto~${\cal B}_N^c$.
\end{proof}

\begin{theorem}\label{th:Af.AN.Afc}\it
Let ${\cal A}=(A,\sigma^A,\delta^A,\tau^A)$ be a fuzzy automaton over an
alphabet $X=\{x_1,\ldots ,x_m\}$ and let $\varphi $~be a weakly right invariant fuzzy quasi-order~on~$A$.~Then
the  automaton ${\cal A}_\varphi^c$ is a homomorphic~image~both~of~${\cal A}_\varphi$~and~${\cal A}_N^c$, and consequently, $|{\cal A}_\varphi^c|\leqslant |{\cal A}_\varphi|$ and $|{\cal A}_\varphi^c|\leqslant |{\cal A}_N^c|$.
\end{theorem}

\begin{proof}
According to Theorems \ref{th:det.red} and \ref{th:Afc.afters},
${\cal A}_\varphi$ is isomorphic to ${\cal B}_N$, and ${\cal A}_\varphi^c$ is~isomorphic to ${\cal B}_N^c$,~and by Theorem 3.4 \cite{JIC.11},
${\cal B}_N^c$ is a homomorphic image of ${\cal B}_N$.~Therefore, ${\cal A}_\varphi^c$ is a homomorphic~image~of~${\cal A}_\varphi$.

On the other hand, ${\cal A}_N^c$ is isomorphic to~${\cal A}_\Delta^c$,~where $\Delta $ is the crisp equality on $A$, and
by Theorem \ref{th:Afc.Afc}, ${\cal A}_\varphi^c$ is a homomorphic image of ${\cal A}_N^c$.
\end{proof}

\section{Algorithms and computational examples}\label{sec:alg}

Let $c_\vee$, $c_\wedge$, $c_\otimes $ and $c_\to $ be respectively computation times of the operations $\vee $, $\wedge $, $\otimes $ and $\to $ in $\cal L$.~In~particular, if $\cal L$ is linearly ordered, we can assume that $c_\vee=c_\wedge=1$, and when $\cal L$ is the G\"odel structure, we can also assume that $c_\otimes =c_\to=1$.

\begin{algorithm}[{{\emph{Construction of the automaton ${\cal A}_\varphi$}}}]\label{alg:A.phi}\rm The input of this algorithm are a fuzzy finite automa\-ton ${\cal A}=(A,\sigma^A,\delta^A,\tau^A)$ with $n$ states, over a finite alphabet $X$ with $m$ letters, and a fuzzy
relation~$\varphi $~on~$A$,~and the output is the crisp-deterministic fuzzy
auto\-maton ${\cal A}_\varphi=(A_\varphi,\varphi_\varepsilon,\delta_\varphi,\tau_\varphi)$.

The procedure is to construct the \emph{transition tree} of ${\cal A}_\varphi$ directly from ${\cal A}$, and during this procedure we use pointers $s(\cdot)$ which points vertices of the tree under construction to the corresponding integers.~The~transition tree of ${\cal A}_\varphi$ is constructed inductively as follows:
\begin{itemize}
\item[(A1)] The root of the tree is $\varphi_\varepsilon=\sigma^A\circ \varphi$, and we put $T_0=\{\varphi_\varepsilon\}$ and
$s(\varphi_\varepsilon)=1$,
and we compute the value $\tau_\varphi(\varphi_\varepsilon)=\varphi_\varepsilon\circ \tau^A$.
\item[(A2)] After the $i$th step let a tree $T_i$ have been constructed,
and vertices in $T_i$ have been labelled either 'closed' or 'non-closed'.
The meaning of these two terms will be made clear in the sequel.
\item[(A3)] In the next step we construct a tree $T_{i+1}$ by enriching $T_i$ in the following way: for any non-closed~leaf $\varphi_u$ occuring~in $T_i$, where $u\in X^*$, and any $x\in X$ we add a vertex $\varphi_{ux}=\varphi_{u}\circ\delta_x^A\circ \varphi $ and an~edge~from~$\varphi_u$ to~$\varphi_{ux}$ la\-belled~by $x$.~Simultaneously, we check whether $\varphi_{ux}$ is a fuzzy set that has already been~constructed.~If~it~is~true,~if $\varphi_{ux}$ is equal to some previously computed $\varphi_v$, we mark $\varphi_{ux}$ as closed and~set $s(\varphi_{ux})=s(\varphi_{v})$.~Otherwise, we compute the value $\tau_\varphi(\varphi_{ux})=\varphi_{ux}\circ \tau^A$~and set $s(\varphi_{ux})$ to~be~the~next~unassigned integer.~The procedure terminates when all leaves are marked~closed.
\item[(A4)] When the transition tree of ${\cal A}_\varphi$ is constructed, we erase
all closure marks and glue leaves to interior vertices with the same pointer value.~The diagram that results is the transition graph of ${\cal A}_\varphi$.
\end{itemize}
\end{algorithm}

When $\varphi $ is taken to be the crisp equality on $A$, Algorithm \ref{alg:A.phi} gives the Nerode automaton ${\cal A}_N$ of $\cal A$.

The above described procedure does not necessarily terminate in a finite number of steps, since~the collection $\{\varphi_u\}_{u\in X^*}$ may be infinite.~However, in cases when this collection is finite, the procedure will terminate in a~finite number of steps, after computing all its members.~For instance, this holds if the~subsemi\-ring
${\cal L}^*(\delta^A,\sigma^A,\varphi)$~of~${\cal L}^*$~gen\-er\-ated by all membership values taken by $\delta^A$, $\sigma^A$ and $\varphi $ is finite (but not~only~in this case). If $k$ denotes~the number of elements of this subsemiring, then~the collection $\{\varphi_u\}_{u\in X^*}$ can have at most $k^n$ different members.

The tree that is constructed by this algorithm is a full $m$-ary tree.~At the end of the algorithm, the~tree can contain at most $k^n$ internal vertices,~and~according
to the well-known theorem on full $m$-ary trees,~the total number of vertices is at most $mk^n+1$.~In~the construction of any single vertice we can first perform
a composition of the form $\varphi_{u}\circ\delta_x^A$, and then a composition
of the form $(\varphi_{u}\circ\delta_x^A)\circ \varphi $, both of which have computation time $O(n^2(c_{\otimes}+c_{\vee}))$.~Therefore,  the computation time for all performed compositions~is $O(mn^2k^{n}(c_{\otimes}+c_{\vee}))$.~Moreover, the tree $T$ has at most $mk^n$ edges, and the computation time of their forming is $O(mk^{n})$.

The time-consuming
part of the~proce\-dure is the check whether the just computed fuzzy set is a copy of some previously computed fuzzy set.~After we have constructed the $j$th fuzzy set, for~some $j\in \Bbb N$ such that $2\leqslant j\leqslant mk^n+1$, we compare it with the~previously constructed fuzzy sets which correspond
to non-closed vertices, whose number is at most $\min\{j-1,k^n\}$.~Therefore, the total number of performed~checks~does not exceed $1+2+\cdots +k^n+(m-1)k^n\cdot k^n=\tfrac12 k^n(k^n+1)+(m-1)k^{2n}$.~As the computation time of any single
check is $O(n)$, the~computation time for all performed checks is $O(mnk^{2n})$.~Summarizing all the above we conclude that the computation time of the
whole algorithm is $O(mnk^{2n})$, the same as~the~computation time of the part in
which for any newly-constructed fuzzy set we check whether it is a copy of
some previously computed fuzzy set.

Note that the number $k$ is characteristic of the fuzzy finite automaton $\cal A$, and it is not a general~characteristic of the semiring ${\cal L}^*$ and its finitely generated subsemirings.~However, if we consider fuzzy automata over a finite lattice, then we can assume that $k$ is the number of elements of this lattice.~Moreover, if~$\cal L$~is the G\"odel structure, then the set of all membership values taken by the fuzzy relations $\{\delta_x^A\}_{x\in X}$ and $\varphi $, and~the fuzzy set $\sigma^A$ is a subsemiring of ${\cal L}^*$, and the number of these values does not exceed $mn^2+n$, so we can use that number instead of $k$.

In a similar way we can provide the following algorithm which constructs the automaton ${\cal A}^\psi$, for some fuzzy relation $\psi $ on $A$, and analyze its computation time.

\begin{algorithm}[{{\emph{Construction of the automaton ${\cal A}^\psi$}}}]\label{alg:A.psi}\rm The input of this algorithm is a fuzzy finite automa\-ton ${\cal A}=(A,\sigma^A,\delta^A,\tau^A)$ with $n$ states, over a finite alphabet $X$ with $m$ letters, and a fuzzy
relation~$\psi $~on~$A$,~and the output is the crisp-deterministic fuzzy
auto\-maton ${\cal A}^\psi=(A^\psi,\psi^\varepsilon,\delta^\psi,\tau^\psi)$.

The procedure is to construct the \emph{transition tree} of ${\cal A}^\psi$ directly from ${\cal A}$, and during this procedure we use pointers $s(\cdot)$ which points vertices of the tree under construction to the corresponding integers.~The~transition tree of of ${\cal A}^\psi$  is constructed inductively as follows:
\begin{itemize}
\item[(A1)] The root of the tree is $\psi^\varepsilon=\psi\circ \tau^A$, and we put $T_0=\{\psi^\varepsilon\}$ and $s(\psi^\varepsilon)=1$, and we compute the value $\tau^\psi(\psi^\varepsilon)=\sigma^A\circ \psi^\varepsilon$.
\item[(A2)] After the $i$th step let a tree $T_i$ have been constructed, and vertices in $T_i$ have been labelled either 'closed' or 'non-closed'. The meaning of these two terms will be made clear in the sequel.
\item[(A3)] In the next step we construct a tree $T_{i+1}$ by enriching $T_i$ in~the following way: for any non-closed~leaf $\psi^u$~occuring~in $T_i$, where $u\in X^*$, and each $x\in X$ we add a vertex $\psi^{xu}=\psi\circ\delta_x^A\circ \psi^{u}$ and an~edge~from~$\psi^u$ to~$\psi^{xu}$ la\-belled~by $x$.~Simultaneously, we check whether $\psi^{xu}$ is a fuzzy set that has already been~constructed.~If~it~is
true,~if $\psi^{xu}$ is equal to some previously computed $\psi^v$, we mark $\psi^{xu}$ as closed~and~set $s(\psi^{xu})=s(\psi^{v})$.~Otherwise, we compute the value $\tau^\psi(\psi^{xu})=\sigma^A\circ \psi^{xu}$~and set $s(\psi^{xu})$ to be the next~unassigned integer.~The procedure terminates when all leaves are marked~closed.
\item[(A4)] When the transition tree of ${\cal A}^\psi$ is constructed, we erase all closure marks and glue leaves to interior vertices with the same pointer value.~The diagram that results is the transition graph of ${\cal A}^\psi$.
\end{itemize}
\end{algorithm}

When $\psi $ is the crisp equality on $A$, Algorithm \ref{alg:A.psi} produces the reverse Nerode automaton ${\cal A}_{\overline N}$ of $\cal A$.

The conditions under which the above procedure terminates in a finite number of steps and its compu\-tation time can be analyzed analogously as in Algorithm
\ref{alg:A.phi}.~The only difference is that instead of the~sub\-semi\-ring
${\cal L}^*(\delta^A,\sigma^A,\varphi)$ here we consider the subsemi\-ring
${\cal L}^*(\delta^A,\tau^A,\psi)$ of ${\cal L}^*$ generated~by~all~member\-ship values taken by $\delta^A$, $\tau^A$ and $\psi $.

\smallskip

As we have said earlier, algorithms for computing the greatest right and left invariant fuzzy \text{quasi-orders} and the greatest weakly right and left invariant fuzzy quasi-orders on a fuzzy finite automaton~were~provided in \cite{SCI.14}. Here we present these algorithms and we perform an analysis of their computation time.

\begin{algorithm}[{{\emph{Computation of the greatest right invariant fuzzy quasi-order}}}]\label{alg:gri}\rm The input of this algorithm is a fuzzy finite
automaton ${\cal A}=(A,\sigma^A,\delta^A,\tau^A)$ with $n$ states, over a finite alphabet $X$ with $m$ letters.~The algorithm computes the greatest right invariant fuzzy~quasi-order $\varphi^{\mathrm{ri}}$ on $\cal A$.

The procedure  constructs the sequence of fuzzy quasi-orders $\{\varphi_k\}_{k\in \Bbb N}$, in the~following way:
\begin{itemize}
\item[(A1)] In the first step we set $\varphi_1=\tau^A/\tau^A$.
\item[(A2)] After the $k$th step let  $\varphi_k$  be the fuzzy quasi-order  that has been constructed.
\item[(A3)] In the next step we construct the fuzzy quasi-order $\varphi_{k+1}$ by means of the formula
    \begin{equation}\label{eq:phi.k.plus.1}
    \varphi_{k+1}=\varphi_k\wedge \bigl[ \bigwedge_{x\in X}(\delta_x^A\circ \varphi_k)/(\delta_x^A\circ \varphi_k)\bigr].
    \end{equation}
\item[(A4)] Simultaneously, we check whether $\varphi_{k+1}=\varphi_k$.
\item[(A5)] When we find the smallest number $s$ such that $\varphi_{s+1}=\varphi_s$, the procedure of constructing the sequence $\{\varphi_k\}_{k\in \Bbb N}$ terminates and $\varphi^{\mathrm{ri}}=\varphi_s$.
\end{itemize}
\end{algorithm}
If the subalgebra ${\cal L}(\delta^A,\tau^A)$ of $\cal L$, generated by all membership values taken by $\delta^A$ and $\tau^A$, satisfies DCC, the algorithm terminates in a finite number of steps.

Consider the computation time of this algorithm.~In (A1) we compute $\tau^A/\tau^A$, which can be done~in~time $O(n^2c_{\to})$.~In (A3) we first compute all compositions  $\delta_x^A\circ \varphi_k$, and if these
computations are performed~according to the definition of composition of fuzzy relations, their computation time is $O(mn^3(c_{\otimes}+c_{\vee}))$.~Then we compute $\varphi_{k+1}$ by means of (\ref{eq:phi.k.plus.1}), and the computation time of this part is $O(mn^3(c_{\to}+c_{\wedge}))$. Thus, the total computation time of (A3) is $O(mn^3(c_{\to}+c_{\wedge}+c_{\otimes}+c_{\vee}))$.~In (A4),
the computation time to check whether $\varphi_{k+1}=\varphi_k$ is $O(n^2)$.

The hardest problem is to estimate the number of steps, in the case when it is finite.~Consider fuzzy relations $\varphi_k$ as fuzzy matrices.~After each step in the construction of the sequence $\{\varphi_k\}_{k\in \Bbb N}$ we check whether some entry has changed its value, and the algorithm terminates after the first step in which there was no change.~Suppose that ${\cal L}(\delta^A,\tau^A)$~satisfies DCC.~Then $\{\{\varphi_k(a,b)\}_{k\in \Bbb N}\mid (a,b)\in A^2\}$ is a finite collection of finite sequences,~so there exists $s\in \Bbb N$ such that the number of different elements in each of these sequences is less than or equal to $s$.~As the sequence $\{\varphi_k\}_{k\in \Bbb N}$ is descending, each entry can change its value at most $s-1$ times, and~the total number of changes is less than or equal to $(s-1)(n^2-n)$ (the diagonal values must~always~be~$1$). Therefore, the algorithm~ter\-mi\-nates after at most $(s-1)(n^2-n)+2$ steps
(in the first and last step values do not change).

Summing up, we get that the total computation time for the whole algorithm is $O(smn^5(c_{\to}+c_{\wedge}+c_{\otimes}+c_{\vee}))$, and hence, the algorithm is polynomial-time.

Let us note that the number $s$ is characteristic of the sequence $\{\varphi_k\}_{k\in
\Bbb N}$, and in general it is not characteristic of the algebra ${\cal L}(\delta^A,\tau^A)$.~However, in some cases the number of different elements in all descending chains in ${\cal L}(\delta^A,\tau^A)$ may have an upper bound $s$.~For example,
if the algebra ${\cal L}(\delta^A,\tau^A)$ is finite, then we~can assume that $s$ is the number of elements of this algebra.~In particular, if $\cal L$ is the G\"odel structure, then the only~values that can be taken by fuzzy relations $\{\varphi_k\}_{k\in \Bbb N}$ are $1$ and those taken by $\delta^A$ and $\tau^A$.~In this case, if $j$ is the number of all values taken by  $\delta^A$ and $\tau^A$, then the
algorithm~ter\-mi\-nates after at most $j(n^2-n)+2$ steps, and total
computation time  is $O(jmn^5)$. Since $j\leqslant mn^2+n^2$, total computation time can also be roughly expressed as $O(m^2n^7)$.

\begin{algorithm}[{{\emph{Computation of the greatest left invariant fuzzy quasi-order}}}]\label{alg:gli}\rm The input of this algorithm is a fuzzy finite
automaton ${\cal A}=(A,\sigma^A,\delta^A,\tau^A)$ with $n$ states, over a finite alphabet $X$ with $m$ letters.~The~algorithm computes the greatest left invariant fuzzy~quasi-order $\psi^{\mathrm{li}}$ on $\cal A$.

The procedure  constructs the sequence of fuzzy quasi-orders $\{\psi_k\}_{k\in \Bbb N}$, in the~following way:
\begin{itemize}
\item[(A1)] In the first step we set $\psi_1=\sigma^A\backslash \sigma^A$.
\item[(A2)] After the $k$th step let $\psi_k$  be the fuzzy quasi-order  that has been constructed.
\item[(A3)] In the next step we construct the fuzzy quasi-order $\psi_{k+1}$ by means of the formula
    \[
    \psi_{k+1}=\psi_k\wedge \bigl[ \bigwedge_{x\in X}(\psi_k\circ \delta_x^A)\backslash (\psi_k\circ \delta_x^A)\bigr].
    \]
\item[(A4)] Simultaneously, we check whether $\psi_{k+1}=\psi_k$.
\item[(A5)] When we find the smallest number $s$ such that $\psi_{s+1}=\psi_s$, the procedure of constructing the sequence $\{\psi_k\}_{k\in \Bbb N}$ terminates and $\psi^{\mathrm{li}}=\psi_s$.
\end{itemize}
\end{algorithm}
The conditions under which this procedure terminates in a finite number of steps and its computation time can be analyzed analogously as in Algorithm \ref{alg:gri}.

\begin{algorithm}[{{\emph{Computation of the greatest weakly right invariant fuzzy quasi-order}}}]\label{alg:gwri}\rm The input of this algorithm is a fuzzy finite automaton ${\cal A}=(A,\sigma^A,\delta^A,\tau^A)$ with $n$ states, over a finite alphabet $X$ with $m$ letters.~The algorithm computes the greatest weakly right invariant fuzzy quasi-order
$\varphi^{\mathrm{wri}}$ on $\cal A$.

The procedure consists of two parts:
\begin{itemize}
\item[(A1)] First we compute all members of the family $\{\tau_u^A\mid u\in X^*\}$, using Algorithm \ref{alg:A.psi}.
\item[(A2)] Then we compute $\varphi^{\mathrm{wri}}$ by means of formula \[
\varphi^{\mathrm{wri}}=\bigwedge_{u\in X^*} \tau_u^A/\tau_u^A .
\]
\end{itemize}
\end{algorithm}
Clearly, this procedure terminates in a finite number of steps under the same conditions as Algorithm~\ref{alg:A.psi}. Under these conditions the computation time of the part (A1) is $O(mnk^{2n})$, and since it dominates over the computation time of (A2), which is $O(k^{n}c_\land+n^2c_\to)$, we conclude that the computation time of the whole algorithm is $O(mnk^{2n})$, the same as for Algorithms \ref{alg:A.phi} and \ref{alg:A.psi}.

Analogous analysis can be performed for the following algorithm that computes the greatest weakly left invariant fuzzy quasi-order on a fuzzy automaton.

\begin{algorithm}[{{\emph{Computation of the greatest weakly left invariant fuzzy quasi-order}}}]\label{alg:gwli}\rm The input of this algorithm is a fuzzy finite automaton ${\cal A}=(A,\sigma^A,\delta^A,\tau^A)$ over a finite alphabet $X$.~The algorithm computes the greatest weakly left invariant fuzzy quasi-order
$\psi^{\mathrm{wli}}$ on $\cal A$.

The procedure consists of two parts:
\begin{itemize}
\item[(A1)] First we compute all members of the family $\{\sigma_u^A\mid u\in X^*\}$, using Algorithm \ref{alg:A.phi}.
\item[(A2)] Then we compute $\psi^{\mathrm{wli}}$ by means of formula \[
\psi^{\mathrm{wli}}=\bigwedge_{u\in X^*} \sigma_u^A\backslash \sigma_u^A .
\]
\end{itemize}
\end{algorithm}

\smallskip

Finally, we turn to the construction of the children automaton ${\cal A}_\varphi^c$.

\begin{algorithm}[{{\emph{Construction of the children automaton ${\cal A}_\varphi^c$}}}]\label{alg:A.phi.c}\rm~The input of this algorithm is a fuzzy~\text{finite} automaton ${\cal A}=(A,\sigma^A,\delta^A,\tau^A)$ with $n$ states, over a finite alphabet $X=\{x_1,\ldots,x_m\}$ with $m$ letters,
and the output is the children automaton ${\cal A}_\varphi^c=(A_\varphi^c,\varphi_\varepsilon^c,\delta_\varphi^c,\tau_\varphi^c)$, where $\varphi $ is the greatest weakly right invariant~fuzzy quasi-order or the greatest right invariant fuzzy quasi-order on $\cal A$.

The procedure is to  construct simultaneously the \emph{transition tree}
of ${\cal A}_\varphi$ and the {\it transition graph\/}~of~${\cal A}_\varphi^c$
directly
from ${\cal A}$.~Except the pointer $s(\cdot)$ used in Algorithm \ref{alg:A.phi},
we also use  another pointer $t(\cdot)$ which~points vertices of the transition graph under construction to the corresponding integers.~The transition tree of~${\cal A}_\varphi$ and the transition graph of ${\cal A}_\varphi^c$ are constructed in~the following way:
\begin{itemize}
\item[(A1)] We compute $\varphi $ using one of Algorithms \ref{alg:gwri} and \ref{alg:gri}, and construct the transition tree $T$ of ${\cal A}_\varphi $ using Algorithm \ref{alg:A.phi}.
\item[(A2)] To each non-closed vertex $\varphi_u$ of the tree $T$ we assign a vertex
$\varphi_u^c$ of a graph $G$ as follows: When~all~children $\varphi_{ux_1},
\ldots, \varphi_{ux_m}$ of $\varphi_u$ in the tree $T$ are formed, we form the vertex $\varphi_u^c=(\varphi_{ux_1},\ldots ,\varphi_{ux_m},\varphi_u\circ
\tau^A)$~in~the graph $G$.~Simultaneously, we check whether $\varphi_{u}^c$
is an $(m+1)$-tuple that has already been constructed.~If it is true, if $\varphi_{u}^c$ is equal to some previously computed $\varphi_v^c$, we mark $\varphi_{u}^c$ as closed~and~set~$t(\varphi_{u}^c)=t(\varphi_{v}^c)$. Otherwise,
we put $\tau_\varphi^c(\varphi_u^c)=\varphi_u\circ \tau^A$ and set $t(\varphi_u^c)$ to be the next integer that has not been used~as~a~value
for $t(\cdot)$.
\item[(A3)] For each non-closed vertex $\varphi_u^c$ of the graph  $G$ and each $x\in X$, if $\varphi_v$
is a non-closed vertex in $T$~such~that $s(\varphi_{ux})=s(\varphi_v)$,~in
the graph $G$ we add an edge from $\varphi_u^c$ to $\varphi_v^c$ labelled by $x$.
\item[(A4)] When the graph $G$ is constructed, we glue closed vertices to non-closed vertices with the same~point\-er value, and erase closure marks.~The diagram that results is the transition graph of ${\cal A}_\varphi^c$.
\end{itemize}
\end{algorithm}
According to Theorems \ref{th:det.red} and \ref{th:Afc.afters}, ${\cal A}_\varphi
$ is isomorphic to ${\cal B}_N$, where ${\cal B}={\cal A}/\varphi $ is the
afterset fuzzy~automa\-ton of $\cal A$ with respect to $\varphi $, and  ${\cal A}_\varphi^c$ is isomorphic to ${\cal B}_N^c$, and by Theorem 3.7 \cite{JIC.11},
${\cal B}_N^c$ is finite if and~only if ${\cal B}_N$ is finite.~Therefore,
${\cal A}_\varphi^c$ is finite if and only if ${\cal A}_\varphi$ is finite.

As in the analysis of Algorithm \ref{alg:A.phi}, consider the case when the
subalgebra ${\cal L}(\delta^A,\sigma^A,\tau^A)$ of $\cal L$ is finite~and
has $l$ elements.~As we have already seen, the computation time of the
part (A1) is  $O(mnl^{2n})$.

When in (A2) we construct a vertex of the graph $G$, as an $(m+1)$-tuple, all its components have~already been computed during construction of the tree $T$, and all that remains to do is to point to them (or~to~their addresses). Hence,
the computation time of forming every single vertex of the graph $G$ is
$O(m)$, and since there are at most $l^n$ vertices in $G$, the computation time of forming all vertices of~$G$~is~$O(ml^n)$.~During~con\-struc\-tion of the tree $T$ we also found which vertices in $T$ are equal as fuzzy sets, and they received the~same pointer values.~Therefore, when we check equality
of two $(m+1)$-tuples in $A_\varphi^c$ we only need to check the equality
of pointer values assigned to their components, and the computation time of such checking~is $O(m)$.~As the total number of checks performed in (A2)
does not exceed $1+2+\cdots +(l^n-1)=\tfrac12 (l^n-1)l^n$, we have that computation time of all checks performed in (A2) is $O(ml^{2n})$.~Finally, in (A3) we form at~most $ml^n$ edges in the graph $G$, and the computation time of this part~is~$O(ml^{n})$.

Therefore, the most expensive part of this algorithm
is (A1), and the computation time of the whole algorithm is  $O(mnl^{2n})$,
the same as for (A1), i.e., the same as for Algorithm 4.7.

\smallskip

Finally we give a remark regarding the computation time of the Brzozowski type algorithm for fuzzy finite automata.~Let ${\cal A}=(A,\sigma^A,\delta^A,\tau^A)$ be a fuzzy finite automaton with $n$ states and $m$ input letters, and suppose~that~the subsemiring ${\cal L}^*(\delta^A,\sigma^A,\tau^A )$ of the semiring ${\cal L}^*$, generated by all membership~values~taken by $\delta^A $,~$\sigma^A $ and~$\tau^A $, is finite and has $k$ elements.~The first round of the application~of~the~Brzozowski~type procedure to~$\cal A$ produces the reverse Nerode automaton of $\cal A$ having at most $k^n$ states, and the computation time of this round is $O(mnk^{2n})$.~The second round may start from an exponentially larger automaton, but despite that, this round produces a minimal crisp-deterministic fuzzy automaton equivalent to $\cal A$, an auto\-maton that is not greater than the Nerode automaton of $\cal A$, which can not have more than $k^n$ states. Thus,~the resulting transition tree can not have more than~$k^n$~internal~vertices, and the  total number of vertices is not greater than $mk^n+1$.~However, computation of any single vertex may be considerably more expensive than in previous algorithms because here we multiply
vectors of size $r$ and matrices of size $r\times r$, where $r\leqslant k^n$.~Therefore, computation of any single vertex of the transition tree in the second round requires time $O(k^{2n}(c_\otimes+c_\lor))$, and the computation time for all vertices is $O(mk^{3n}(c_\otimes+c_\lor))$.
Since the tree has at most $mk^n$ edges, the computation time of their forming
is $O(mk^n)$.

When for any newly-constructed fuzzy set we check whether it is a copy of some previously computed fuzzy set, the total number of performed checks
is $\frac 12 k^n(k^n+1)+(m-1)k^{2n}$, the same as in previous~algorithms, but here any single check has the computation time $O(k^{n})$,
so the computation time for all performed checks is $O(mk^{3n})$.~Hence, the computation time of the whole algorithm is $O(mk^{3n}(c_\otimes+c_\lor))$,
or  $O(mk^{3n})$,~if~the~operations $\otimes $ and $\lor $ can be performed in constant time.~Accordingly, the Brzozowski type algorithm is~somewhat slower than the other algorithms discussed here, but its performances can be improved if instead of the construction of the reverse Nerode automaton~we~use~the~con\-struction of the automaton corresponding to the greatest right invariant or weakly right invariant fuzzy quasi-order on $\cal A$, or the construction of its children automaton.

\smallskip

Now we provide several illustrative computational examples.

\begin{example}\label{ex:N.ri.wri}\rm
Let ${\cal A}$ be a Boolean automaton over the two-element alphabet $X=\{x,y\}$ given by the transition graph shown in Fig.~\ref{fig:N.ri.wri.0}.
\begin{figure}
%%%%
\begin{center}
%%%%
\psset{unit=1cm}
\newpsobject{showgrid}{psgrid}{subgriddiv=1,griddots=10,gridlabels=6pt}
%%%
\begin{pspicture}(-1,0.3)(4,3.4)%\showgrid
\pnode(0,2){AP1}
\SpecialCoor
\rput(AP1){\cnode{2.7mm}{AP1A1}}
\rput(AP1A1){\scriptsize$a_1$}
\rput([angle=-30,nodesep=25mm,offset=0pt]AP1A1){\cnode[doubleline=true]{3mm}{AP1A2}}
\rput(AP1A2){\scriptsize$a_2$}
\rput([angle=30,nodesep=25mm,offset=0pt]AP1A1){\cnode[doubleline=true]{3mm}{AP1A3}}
\rput(AP1A3){\scriptsize$a_3$}
\rput([angle=180,nodesep=5mm,offset=0pt]AP1A1){\pnode{AP1I}}
\ncline{->}{AP1I}{AP1A1}
\ncarc[arcangle=16]{<-}{AP1A1}{AP1A2}\aput[1pt](.5){\scriptsize $x,y$}
\ncarc[arcangle=16]{<-}{AP1A2}{AP1A1}\aput[1pt](.5){\scriptsize $x$}
\ncarc[arcangle=16]{<-}{AP1A1}{AP1A3}\aput[1pt](.5){\scriptsize $x$}
\ncarc[arcangle=16]{<-}{AP1A3}{AP1A1}\aput[1pt](.5){\scriptsize $y$}
\ncarc[arcangle=16]{<-}{AP1A2}{AP1A3}\aput[1pt](.5){\scriptsize $y$}
\ncarc[arcangle=16]{<-}{AP1A3}{AP1A2}\aput[1pt](.5){\scriptsize $x$}
%
%\nccircle[angleA=-90]{<-}{AP1A2}{0.4}\bput[0pt](.40){\scriptsize $y$}
\nccurve[angleA=45,angleB=-45,ncurv=4]{->}{AP1A2}{AP1A2}\aput[0.5pt](.50){\scriptsize $y$}
\NormalCoor
\end{pspicture}
%%%%
\end{center}
\caption{\scriptsize The transition graph of the fuzzy automaton ${\cal A}$ from Example \ref{ex:N.ri.wri}.}\label{fig:N.ri.wri.0}
%%%%
\end{figure}
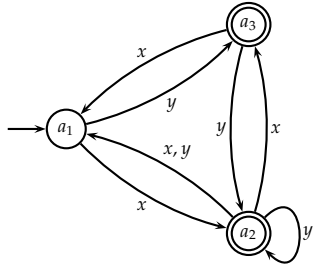
%%%%

The transition tree and the transition graph of the Nerode automaton ${\cal
A}_N$ of $\cal A$, constructed by means~of Algorithm \ref{alg:A.phi}, are presented in Fig.~\ref{fig:N.ri.wri.1}.~We see that the Nerode automaton
${\cal A}_N$ has 7 states.

By means of Algorithms \ref{alg:gri} and \ref{alg:gwri} we compute the gratest
right invariant fuzzy quasi-order~$\varphi^{\textrm{ri}}$~and~the gratest
weakly right invariant fuzzy quasi-order $\varphi^{\textrm{wri}}$ on $\cal
A$, which are represented by the following Boolean matrices:
\[
\varphi^{\textrm{ri}}=\begin{bmatrix} 1 & 0 & 0 \\ 0 & 1 & 1 \\ 0 & 0 & 1 \end{bmatrix}, \qquad \varphi^{\textrm{wri}}=\begin{bmatrix} 1 & 0 & 0 \\ 1 & 1 & 1 \\ 0 & 0 & 1 \end{bmatrix}.
\]
Then, using Algorithm \ref{alg:A.phi} we construct automata ${\cal A}_{\varphi^{\textrm{ri}}}$ and ${\cal A}_{\varphi^{\textrm{wri}}}$, whose transition trees and graphs are presented in Fig. \ref{fig:N.ri.wri.2} and \ref{fig:N.ri.wri.3}, respectively.~The automaton ${\cal A}_{\varphi^{\textrm{ri}}}$ has 5 states, whereas the automaton ${\cal A}_{\varphi^{\textrm{wri}}}$~has 3 states.~Finally, using Algorithm \ref{alg:A.phi.c} and the transition tree of ${\cal A}_N$ from Fig.~\ref{fig:N.ri.wri.1}, we construct~the~children auto\-maton ${\cal A}^c_N$ of ${\cal A}_N$ presented in Fig.~\ref{fig:N.ri.wri.3}.~Clearly, ${\cal A}^c_N$ is isomorphic to ${\cal A}_{\varphi^{\textrm{ri}}}$.

\begin{figure}
%%%%
\begin{center}
%%%%
\psset{unit=1cm}
\newpsobject{showgrid}{psgrid}{subgriddiv=1,griddots=10,gridlabels=6pt}%\showgrid
%%%
\begin{pspicture}(-6,-0.2)(11,5)%\showgrid
\rput(-6,5){\textsf{a)}}
\rput(6,5){\textsf{b)}}
\pnode(0,0){C}
\SpecialCoor
\rput(C){\cnode{3.5mm}{SE}}
\rput(SE){\scriptsize$\sigma_\varepsilon$}
\rput([angle=165,nodesep=30mm,offset=0pt]C){\cnode{3.5mm}{SX}}
\rput(SX){\scriptsize$\sigma_x$}
\rput([angle=15,nodesep=30mm,offset=0pt]C){\cnode{3.5mm}{SY}}
\rput(SY){\scriptsize$\sigma_y$}
\rput([angle=150,nodesep=20mm,offset=0pt]SX){\cnode{3.5mm}{SX2}}
\rput(SX2){\scriptsize$\sigma_{\!x^2}$}
\rput([angle=30,nodesep=20mm,offset=0pt]SX){\cnode{3.5mm}{SXY}}
\rput(SXY){\scriptsize$\sigma_{\!xy}$}
\rput([angle=150,nodesep=20mm,offset=0pt]SY){\cnode{3.5mm}{SYX}}
\rput(SYX){\scriptsize$\sigma_{\!yx}$}
\nput[labelsep=1pt]{90}{SYX}{$\blacksquare$}
\rput([angle=30,nodesep=20mm,offset=0pt]SY){\cnode{3.5mm}{SY2}}
\rput(SY2){\scriptsize$\sigma_{\!y^2}$}
\nput[labelsep=1pt]{90}{SY2}{$\blacksquare$}
\rput([angle=120,nodesep=12mm,offset=0pt]SX2){\cnode{3.5mm}{SX3}}
\rput(SX3){\scriptsize$\sigma_{\!x^3}$}
\nput[labelsep=1pt]{90}{SX3}{$\blacksquare$}
\rput([angle=60,nodesep=12mm,offset=0pt]SX2){\cnode{3.5mm}{SX2Y}}
\rput(SX2Y){\scriptsize$\sigma_{\!x^2y}$}
\rput([angle=120,nodesep=12mm,offset=0pt]SXY){\cnode{3.5mm}{SXYX}}
\rput(SXYX){\scriptsize$\sigma_{\!xyx}$}
\rput([angle=60,nodesep=12mm,offset=0pt]SXY){\cnode{3.5mm}{SXY2}}
\rput(SXY2){\scriptsize$\sigma_{\!xy^2}$}
\nput[labelsep=1pt]{90}{SXY2}{$\blacksquare$}
\rput([angle=120,nodesep=10mm,offset=0pt]SX2Y){\cnode{3.5mm}{SX2YX}}
\rput(SX2YX){\scriptsize$\sigma_{\!x^2yx}$}
\nput[labelsep=1pt]{90}{SX2YX}{$\blacksquare$}
\rput([angle=60,nodesep=10mm,offset=0pt]SX2Y){\cnode{3.5mm}{SX2Y2}}
\rput(SX2Y2){\scriptsize$\sigma_{\!x^2y^2}$}
\nput[labelsep=1pt]{90}{SX2Y2}{$\blacksquare$}
\rput([angle=120,nodesep=10mm,offset=0pt]SXYX){\cnode{3.5mm}{SXYX2}}
\rput(SXYX2){\scriptsize$\sigma_{\!xyx^2}$}
\nput[labelsep=1pt]{90}{SXYX2}{$\blacksquare$}
\rput([angle=60,nodesep=10mm,offset=0pt]SXYX){\cnode{3.5mm}{SXYXY}}
\rput(SXYXY){\scriptsize$\sigma_{\!(xy)^2}$}
\nput[labelsep=1pt]{90}{SXYXY}{$\blacksquare$}
%%%
\NormalCoor
\ncline{->}{SE}{SX}\aput[1pt](.6){\scriptsize $x$}
\ncline{->}{SE}{SY}\bput[1pt](.6){\scriptsize $y$}
\ncline{->}{SX}{SX2}\aput[1pt](.6){\scriptsize $x$}
\ncline{->}{SX}{SXY}\bput[1pt](.6){\scriptsize $y$}
\ncline{->}{SY}{SYX}\aput[1pt](.6){\scriptsize $x$}
\ncline{->}{SY}{SY2}\bput[1pt](.6){\scriptsize $y$}
\ncline{->}{SX2}{SX3}\aput[1pt](.5){\scriptsize $x$}
\ncline{->}{SX2}{SX2Y}\bput[1pt](.5){\scriptsize $y$}
\ncline{->}{SXY}{SXYX}\aput[1pt](.5){\scriptsize $x$}
\ncline{->}{SXY}{SXY2}\bput[1pt](.5){\scriptsize $y$}
\ncline{->}{SX2Y}{SX2YX}\aput[1pt](.5){\scriptsize $x$}
\ncline{->}{SX2Y}{SX2Y2}\bput[1pt](.5){\scriptsize $y$}
\ncline{->}{SXYX}{SXYX2}\aput[1pt](.5){\scriptsize $x$}
\ncline{->}{SXYX}{SXYXY}\bput[1pt](.5){\scriptsize $y$}
\ncarc[linestyle=dashed, dash=0.7pt 2pt]{->}{SYX}{SE}
\ncarc[linestyle=dashed, dash=0.7pt 2pt]{->}{SY2}{SX}
\ncarc[linestyle=dashed, dash=0.7pt 2pt]{<-}{SXY}{SX3}
\ncarc[linestyle=dashed, dash=0.7pt 2pt]{->}{SXY2}{SXYX}
\ncarc[linestyle=dashed, dash=0.7pt 2pt]{<-}{SX2}{SX2YX}
\ncarc[linestyle=dashed, dash=0.7pt 2pt]{<-}{SXY}{SX2Y2}
\ncarc[linestyle=dashed, dash=0.7pt 2pt,arcangle=20]{->}{SXYX2}{SXYX}
\ncarc[linestyle=dashed, dash=0.7pt 2pt,arcangle=20]{<-}{SXYX}{SXYXY}
%%%%%%%%%%%
%%%%%%%%%%%
\pnode(9.5,0){C2}
\SpecialCoor
\rput(C2){\cnode{3mm}{GSE}}
\rput(GSE){\scriptsize$\sigma_\varepsilon$}
\rput([angle=180,nodesep=5mm,offset=0pt]GSE){\pnode{GI}}
\ncline{->}{GI}{GSE}
%
%\rput([angle=135,nodesep=15mm,offset=0pt]GSE){\cnode{3.5mm}{GSX}}
\rput([angle=135,nodesep=15mm,offset=0pt]GSE){\cnode[doubleline=true]{3.5mm}{GSX}}
\rput(GSX){\scriptsize$\sigma_x$}
\rput([angle=45,nodesep=15mm,offset=0pt]GSE){\cnode[doubleline=true]{3.5mm}{GSY}}
\rput(GSY){\scriptsize$\sigma_y$}
\rput([angle=135,nodesep=15mm,offset=0pt]GSX){\cnode[doubleline=true]{3.5mm}{GSX2}}
\rput(GSX2){\scriptsize$\sigma_{\!x^2}$}
%
%\rput([angle=45,nodesep=15mm,offset=0pt]GSX){\cnode{3.5mm}{GSXY}}
\rput([angle=45,nodesep=15mm,offset=0pt]GSX){\cnode[doubleline=true]{3.5mm}{GSXY}}
\rput(GSXY){\scriptsize$\sigma_{\!xy}$}
\rput([angle=90,nodesep=13mm,offset=0pt]GSX2){\cnode[doubleline=true]{3.5mm}{GSX2Y}}
\rput(GSX2Y){\scriptsize$\sigma_{\!x^2y}$}
\rput([angle=90,nodesep=13mm,offset=0pt]GSXY){\cnode[doubleline=true]{3.5mm}{GSXYX}}
\rput(GSXYX){\scriptsize$\sigma_{\!xyx}$}
\ncline{->}{GSE}{GSX}\aput[1pt](.5){\scriptsize $x$}
\ncarc[arcangle=16]{<-}{GSY}{GSE}\aput[1pt](.5){\scriptsize $y$}
\ncline{->}{GSX}{GSX2}\aput[1pt](.5){\scriptsize $x$}
\ncline{->}{GSX}{GSXY}\bput[1pt](.5){\scriptsize $y$}
\ncline{->}{GSX2}{GSXY}\bput[1pt](.5){\scriptsize $x$}
\ncarc[arcangle=16]{<-}{GSE}{GSY}\aput[1pt](.5){\scriptsize $x$}
\ncarc{<-}{GSX}{GSY}\aput[1pt](.5){\scriptsize $y$}
\ncarc[arcangle=16]{->}{GSX2}{GSX2Y}\aput[1pt](.5){\scriptsize $y$}
\ncarc[arcangle=16]{->}{GSX2Y}{GSX2}\aput[1pt](.5){\scriptsize $x$}
\ncline{->}{GSXY}{GSXYX}\bput[1pt](.5){\scriptsize $x,y$}
\ncline{->}{GSX2Y}{GSXY}\bput[1pt](.5){\scriptsize $y$}
\nccurve[angleA=135,angleB=45,ncurv=4]{->}{GSXYX}{GSXYX}\aput[0.5pt](.50){\scriptsize $x,y$}
%\nccircle[angleA=90]{->}{GSX2}{0.4}\bput[0pt](.40){\scriptsize $x$}
%\nccircle[angleB=180]{<-}{GSXYX}{0.4}\bput[0pt](.50){\scriptsize $x,y$}
%
\NormalCoor
\end{pspicture}
%%%%
\end{center}
\caption{\scriptsize The transition tree (a)) and the transition graph (b)) of the Nerode automaton ${\cal A}_N$ of the fuzzy automaton ${\cal A}$ from Example \ref{ex:N.ri.wri}.}\label{fig:N.ri.wri.1}
%%%%
\end{figure}
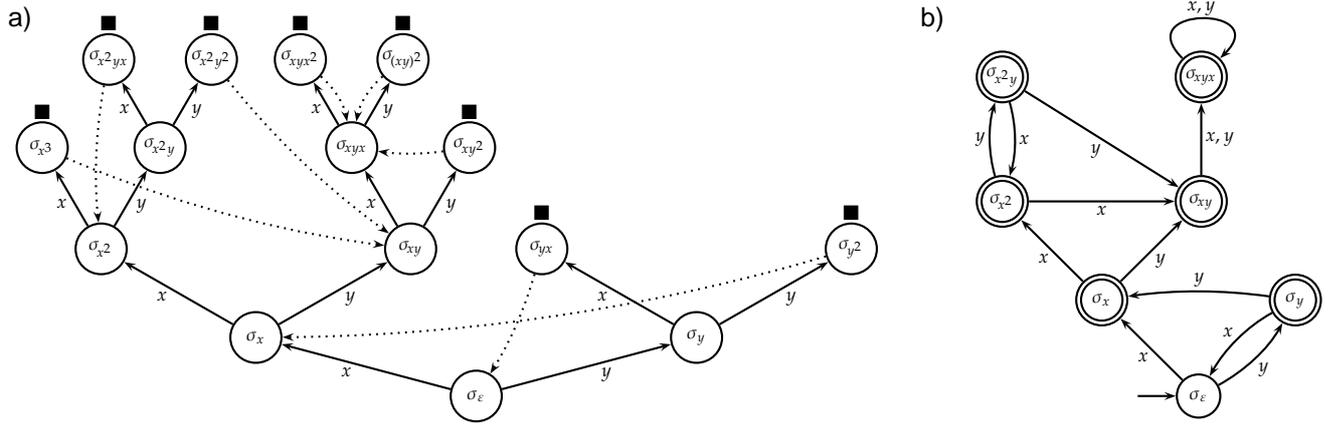
%%%%

\begin{figure}
%%%%
\begin{center}
%%%%
\psset{unit=1cm}
\newpsobject{showgrid}{psgrid}{subgriddiv=1,griddots=10,gridlabels=6pt}%\showgrid
%%%
\begin{pspicture}(-6,-0.2)(11,4)%\showgrid
\rput(-6,4){\textsf{a)}}
\rput(4.5,4){\textsf{b)}}
\pnode(-1,0){C}
\SpecialCoor
\rput(C){\cnode{3.5mm}{SE}}
\rput(SE){\scriptsize$\varphi^{\textrm{ri}}_\varepsilon$}
\rput([angle=160,nodesep=22mm,offset=0pt]C){\cnode{3.5mm}{SX}}
\rput(SX){\scriptsize$\varphi^{\textrm{ri}}_x$}
\rput([angle=20,nodesep=22mm,offset=0pt]C){\cnode{3.5mm}{SY}}
\rput(SY){\scriptsize$\varphi^{\textrm{ri}}_y$}
\rput([angle=135,nodesep=15mm,offset=0pt]SX){\cnode{3.5mm}{SX2}}
\rput(SX2){\scriptsize$\varphi^{\textrm{ri}}_{\!x^2}$}
\rput([angle=45,nodesep=15mm,offset=0pt]SX){\cnode{3.5mm}{SXY}}
\rput(SXY){\scriptsize$\varphi^{\textrm{ri}}_{\!xy}$}
\rput([angle=135,nodesep=15mm,offset=0pt]SY){\cnode{3.5mm}{SYX}}
\rput(SYX){\scriptsize$\varphi^{\textrm{ri}}_{\!yx}$}
\nput[labelsep=1pt]{90}{SYX}{$\blacksquare$}
\rput([angle=45,nodesep=15mm,offset=0pt]SY){\cnode{3.5mm}{SY2}}
\rput(SY2){\scriptsize$\varphi^{\textrm{ri}}_{\!y^2}$}
\nput[labelsep=1pt]{90}{SY2}{$\blacksquare$}
\rput([angle=120,nodesep=12mm,offset=0pt]SX2){\cnode{3.5mm}{SX3}}
\rput(SX3){\scriptsize$\varphi^{\textrm{ri}}_{\!x^3}$}
\nput[labelsep=1pt]{90}{SX3}{$\blacksquare$}
\rput([angle=60,nodesep=12mm,offset=0pt]SX2){\cnode{3.5mm}{SX2Y}}
\rput(SX2Y){\scriptsize$\varphi^{\textrm{ri}}_{\!x^2y}$}
\nput[labelsep=1pt]{90}{SX2Y}{$\blacksquare$}
\rput([angle=120,nodesep=12mm,offset=0pt]SXY){\cnode{3.5mm}{SXYX}}
\rput(SXYX){\scriptsize$\varphi^{\textrm{ri}}_{\!xyx}$}
\nput[labelsep=1pt]{90}{SXYX}{$\blacksquare$}
\rput([angle=60,nodesep=12mm,offset=0pt]SXY){\cnode{3.5mm}{SXY2}}
\rput(SXY2){\scriptsize$\varphi^{\textrm{ri}}_{\!xy^2}$}
\nput[labelsep=1pt]{90}{SXY2}{$\blacksquare$}
\NormalCoor
\ncline{->}{SE}{SX}\aput[1pt](.6){\scriptsize $x$}
\ncline{->}{SE}{SY}\bput[1pt](.6){\scriptsize $y$}
\ncline{->}{SX}{SX2}\aput[1pt](.6){\scriptsize $x$}
\ncline{->}{SX}{SXY}\bput[1pt](.6){\scriptsize $y$}
\ncline{->}{SY}{SYX}\aput[1pt](.6){\scriptsize $x$}
\ncline{->}{SY}{SY2}\bput[1pt](.6){\scriptsize $y$}
\ncline{->}{SX2}{SX3}\aput[1pt](.5){\scriptsize $x$}
\ncline{->}{SX2}{SX2Y}\bput[1pt](.5){\scriptsize $y$}
\ncline{->}{SXY}{SXYX}\aput[1pt](.5){\scriptsize $x$}
\ncline{->}{SXY}{SXY2}\bput[1pt](.5){\scriptsize $y$}
\ncarc[linestyle=dashed, dash=0.7pt 2pt]{->}{SYX}{SE}
\ncarc[linestyle=dashed, dash=0.7pt 2pt]{->}{SY2}{SX}
\ncarc[linestyle=dashed, dash=0.7pt 2pt]{<-}{SXY}{SX3}
\ncarc[linestyle=dashed, dash=0.7pt 2pt]{<-}{SX}{SX2Y}
\ncarc[linestyle=dashed, dash=0.7pt 2pt,arcangle=20]{->}{SXYX}{SXY}
\ncarc[linestyle=dashed, dash=0.7pt 2pt,arcangle=20]{<-}{SXY}{SXY2}
%%%%%%%%%%%
%%%%%%%%%%%
\pnode(8,1){C2}
\SpecialCoor
\rput(C2){\cnode{3mm}{GSE}}
\rput(GSE){\scriptsize$\varphi^{\textrm{ri}}_\varepsilon$}
\rput([angle=180,nodesep=5mm,offset=0pt]GSE){\pnode{GI}}
\ncline{->}{GI}{GSE}
\rput([angle=135,nodesep=15mm,offset=0pt]GSE){\cnode[doubleline=true]{3.5mm}{GSX}}
\rput(GSX){\scriptsize$\varphi^{\textrm{ri}}_x$}
\rput([angle=45,nodesep=15mm,offset=0pt]GSE){\cnode[doubleline=true]{3.5mm}{GSY}}
\rput(GSY){\scriptsize$\varphi^{\textrm{ri}}_y$}
\rput([angle=135,nodesep=15mm,offset=0pt]GSX){\cnode[doubleline=true]{3.5mm}{GSX2}}
\rput(GSX2){\scriptsize$\varphi^{\textrm{ri}}_{\!x^2}$}
\rput([angle=45,nodesep=15mm,offset=0pt]GSX){\cnode[doubleline=true]{3.5mm}{GSXY}}
\rput(GSXY){\scriptsize$\varphi^{\textrm{ri}}_{\!xy}$}
\ncline{->}{GSE}{GSX}\aput[1pt](.5){\scriptsize $x$}
\ncarc[arcangle=16]{<-}{GSY}{GSE}\aput[1pt](.5){\scriptsize $y$}
\ncline{->}{GSX}{GSXY}\bput[1pt](.5){\scriptsize $y$}
\ncarc[arcangle=16]{<-}{GSE}{GSY}\aput[1pt](.5){\scriptsize $x$}
\ncarc{<-}{GSX}{GSY}\aput[1pt](.5){\scriptsize $y$}
\ncarc[arcangle=16]{->}{GSX2}{GSX}\aput[1pt](.5){\scriptsize $y$}
\ncarc[arcangle=16]{->}{GSX}{GSX2}\aput[1pt](.5){\scriptsize $x$}
\ncarc{->}{GSX2}{GSXY}\aput[1pt](.5){\scriptsize $x$}
\nccurve[angleA=135,angleB=45,ncurv=4]{->}{GSXY}{GSXY}\aput[0.5pt](.50){\scriptsize $x,y$}
%\nccircle[angleA=-90]{<-}{GSXY}{0.4}\bput[0pt](.50){\scriptsize $x,y$}
%
\NormalCoor
\end{pspicture}
%%%%
\end{center}
\caption{\scriptsize The transition tree (a)) and the transition graph (b)) of the automaton ${\cal A}_{\varphi^{\textrm{ri}}}$ for the fuzzy automaton ${\cal A}$ from Example \ref{ex:N.ri.wri}.}\label{fig:N.ri.wri.2}
%%%%
\end{figure}
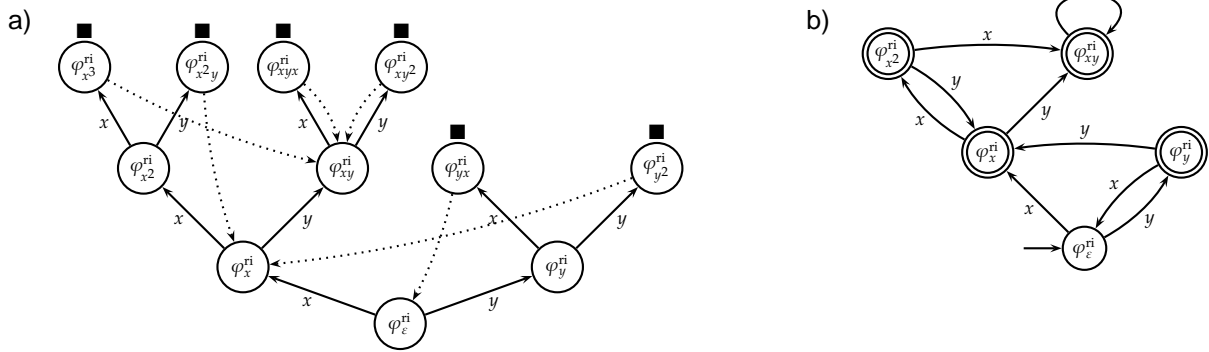
%%%%

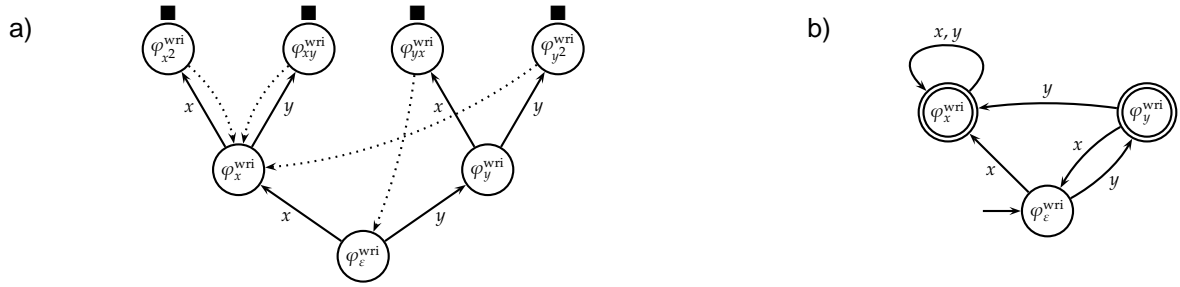
\begin{figure}
%%%%
\begin{center}
%%%%
\psset{unit=1cm}
\newpsobject{showgrid}{psgrid}{subgriddiv=1,griddots=10,gridlabels=6pt}%\showgrid
%%%
\begin{pspicture}(-6,-0.2)(11,3)%\showgrid
\rput(-6,3){\textsf{a)}}
\rput(4.5,3){\textsf{b)}}
\pnode(-1.5,0){C}
\SpecialCoor
\rput(C){\cnode{3.5mm}{SE}}
\rput(SE){\scriptsize$\varphi^{\textrm{wri}}_\varepsilon$}
\rput([angle=145,nodesep=20mm,offset=0pt]C){\cnode{3.5mm}{SX}}
\rput(SX){\scriptsize$\varphi^{\textrm{wri}}_x$}
\rput([angle=35,nodesep=20mm,offset=0pt]C){\cnode{3.5mm}{SY}}
\rput(SY){\scriptsize$\varphi^{\textrm{wri}}_y$}
\rput([angle=120,nodesep=15mm,offset=0pt]SX){\cnode{3.5mm}{SX2}}
\rput(SX2){\scriptsize$\varphi^{\textrm{wri}}_{\!x^2}$}
\nput[labelsep=1pt]{90}{SX2}{$\blacksquare$}
\rput([angle=60,nodesep=15mm,offset=0pt]SX){\cnode{3.5mm}{SXY}}
\rput(SXY){\scriptsize$\varphi^{\textrm{wri}}_{\!xy}$}
\nput[labelsep=1pt]{90}{SXY}{$\blacksquare$}
\rput([angle=120,nodesep=15mm,offset=0pt]SY){\cnode{3.5mm}{SYX}}
\rput(SYX){\scriptsize$\varphi^{\textrm{wri}}_{\!yx}$}
\nput[labelsep=1pt]{90}{SYX}{$\blacksquare$}
\rput([angle=60,nodesep=15mm,offset=0pt]SY){\cnode{3.5mm}{SY2}}
\rput(SY2){\scriptsize$\varphi^{\textrm{wri}}_{\!y^2}$}
\nput[labelsep=1pt]{90}{SY2}{$\blacksquare$}
\NormalCoor
\ncline{->}{SE}{SX}\aput[1pt](.6){\scriptsize $x$}
\ncline{->}{SE}{SY}\bput[1pt](.6){\scriptsize $y$}
\ncline{->}{SX}{SX2}\aput[1pt](.6){\scriptsize $x$}
\ncline{->}{SX}{SXY}\bput[1pt](.6){\scriptsize $y$}
\ncline{->}{SY}{SYX}\aput[1pt](.6){\scriptsize $x$}
\ncline{->}{SY}{SY2}\bput[1pt](.6){\scriptsize $y$}
\ncarc[linestyle=dashed, dash=0.7pt 2pt]{->}{SYX}{SE}
\ncarc[linestyle=dashed, dash=0.7pt 2pt,arcangle=16]{->}{SY2}{SX}
\ncarc[linestyle=dashed, dash=0.7pt 2pt,arcangle=20]{->}{SX2}{SX}
\ncarc[linestyle=dashed, dash=0.7pt 2pt,arcangle=-20]{->}{SXY}{SX}
%%%%%%%%%%%
%%%%%%%%%%%
\pnode(7.5,0.6){C2}
\SpecialCoor
\rput(C2){\cnode{3.5mm}{GSE}}
\rput(GSE){\scriptsize$\varphi^{\textrm{wri}}_\varepsilon$}
\rput([angle=180,nodesep=5mm,offset=0pt]GSE){\pnode{GI}}
\ncline{->}{GI}{GSE}
\rput([angle=135,nodesep=15mm,offset=0pt]GSE){\cnode[doubleline=true]{4mm}{GSX}}
\rput(GSX){\scriptsize$\varphi^{\textrm{wri}}_x$}
\rput([angle=45,nodesep=15mm,offset=0pt]GSE){\cnode[doubleline=true]{4mm}{GSY}}
\rput(GSY){\scriptsize$\varphi^{\textrm{wri}}_y$}
\ncline{->}{GSE}{GSX}\aput[1pt](.5){\scriptsize $x$}
\ncarc[arcangle=16]{<-}{GSY}{GSE}\aput[1pt](.5){\scriptsize $y$}
\ncarc[arcangle=16]{<-}{GSE}{GSY}\aput[1pt](.5){\scriptsize $x$}
\ncarc{<-}{GSX}{GSY}\aput[1pt](.5){\scriptsize $y$}
\nccurve[angleA=135,angleB=45,ncurv=4]{<-}{GSX}{GSX}\aput[0.5pt](.50){\scriptsize $x,y$}
\NormalCoor
\end{pspicture}
%%%%
\end{center}
\caption{\scriptsize The transition tree (a)) and the transition graph (b)) of the automaton ${\cal A}_{\varphi^{\textrm{wri}}}$ for the fuzzy automaton ${\cal A}$ from Example \ref{ex:N.ri.wri}.}\label{fig:N.ri.wri.3}
%%%%
\end{figure}
%%%%

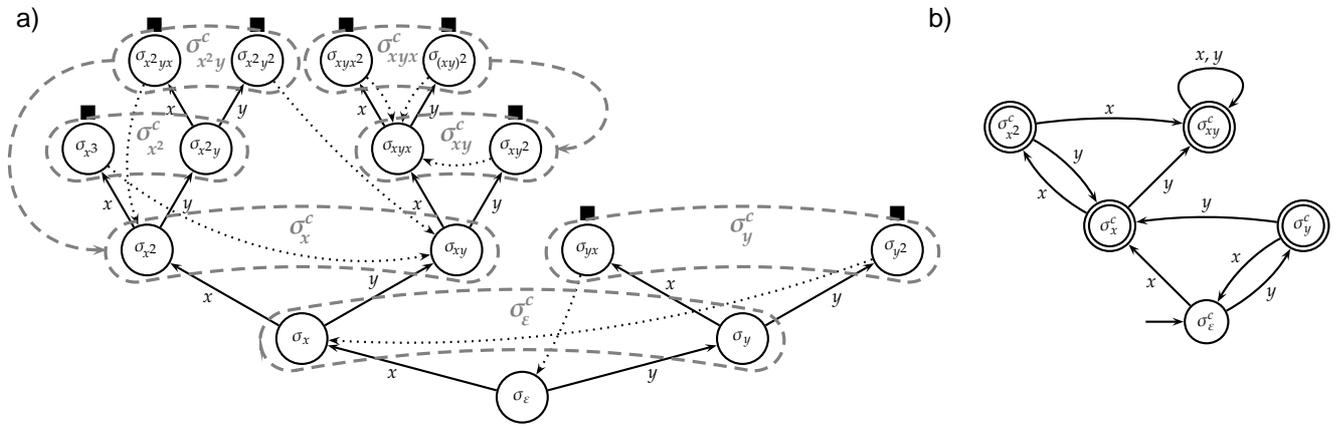
\begin{figure}
%%%%
\begin{center}
%%%%
\psset{unit=1cm}
\newpsobject{showgrid}{psgrid}{subgriddiv=1,griddots=10,gridlabels=6pt}%\showgrid
%%%
\begin{pspicture}(-6.5,-0.2)(10.5,5)%\showgrid
\rput(-6.5,5){\textsf{a)}}
\rput(5.5,5){\textsf{b)}}
\pnode(0,0){C}
\SpecialCoor
\rput(C){\cnode{3.5mm}{SE}}
\rput(SE){\scriptsize$\sigma_\varepsilon$}
\rput([angle=165,nodesep=30mm,offset=0pt]C){\cnode{3.5mm}{SX}}
\rput(SX){\scriptsize$\sigma_x$}
\rput([angle=15,nodesep=30mm,offset=0pt]C){\cnode{3.5mm}{SY}}
\rput(SY){\scriptsize$\sigma_y$}
\rput([angle=150,nodesep=20mm,offset=0pt]SX){\cnode{3.5mm}{SX2}}
\rput(SX2){\scriptsize$\sigma_{\!x^2}$}
\rput([angle=30,nodesep=20mm,offset=0pt]SX){\cnode{3.5mm}{SXY}}
\rput(SXY){\scriptsize$\sigma_{\!xy}$}
\rput([angle=150,nodesep=20mm,offset=0pt]SY){\cnode{3.5mm}{SYX}}
\rput(SYX){\scriptsize$\sigma_{\!yx}$}
\nput[labelsep=1pt]{90}{SYX}{$\blacksquare$}
\rput([angle=30,nodesep=20mm,offset=0pt]SY){\cnode{3.5mm}{SY2}}
\rput(SY2){\scriptsize$\sigma_{\!y^2}$}
\nput[labelsep=1pt]{90}{SY2}{$\blacksquare$}
\rput([angle=120,nodesep=12mm,offset=0pt]SX2){\cnode{3.5mm}{SX3}}
\rput(SX3){\scriptsize$\sigma_{\!x^3}$}
\nput[labelsep=1pt]{90}{SX3}{$\blacksquare$}
\rput([angle=60,nodesep=12mm,offset=0pt]SX2){\cnode{3.5mm}{SX2Y}}
\rput(SX2Y){\scriptsize$\sigma_{\!x^2y}$}
\rput([angle=120,nodesep=12mm,offset=0pt]SXY){\cnode{3.5mm}{SXYX}}
\rput(SXYX){\scriptsize$\sigma_{\!xyx}$}
\rput([angle=60,nodesep=12mm,offset=0pt]SXY){\cnode{3.5mm}{SXY2}}
\rput(SXY2){\scriptsize$\sigma_{\!xy^2}$}
\nput[labelsep=1pt]{90}{SXY2}{$\blacksquare$}
\rput([angle=120,nodesep=10mm,offset=0pt]SX2Y){\cnode{3.5mm}{SX2YX}}
\rput(SX2YX){\scriptsize$\sigma_{\!x^2yx}$}
\nput[labelsep=1pt]{90}{SX2YX}{$\blacksquare$}
\rput([angle=60,nodesep=10mm,offset=0pt]SX2Y){\cnode{3.5mm}{SX2Y2}}
\rput(SX2Y2){\scriptsize$\sigma_{\!x^2y^2}$}
\nput[labelsep=1pt]{90}{SX2Y2}{$\blacksquare$}
\rput([angle=120,nodesep=10mm,offset=0pt]SXYX){\cnode{3.5mm}{SXYX2}}
\rput(SXYX2){\scriptsize$\sigma_{\!xyx^2}$}
\nput[labelsep=1pt]{90}{SXYX2}{$\blacksquare$}
\rput([angle=60,nodesep=10mm,offset=0pt]SXYX){\cnode{3.5mm}{SXYXY}}
\rput(SXYXY){\scriptsize$\sigma_{\!(xy)^2}$}
\nput[labelsep=1pt]{90}{SXYXY}{$\blacksquare$}
\NormalCoor
\ncline{->}{SE}{SX}\aput[1pt](.6){\scriptsize $x$}
\ncline{->}{SE}{SY}\bput[1pt](.6){\scriptsize $y$}
\ncline{->}{SX}{SX2}\aput[1pt](.6){\scriptsize $x$}
\ncline{->}{SX}{SXY}\aput[1pt](.5){\scriptsize $y$}
\ncline{->}{SY}{SYX}\bput[1pt](.5){\scriptsize $x$}
\ncline{->}{SY}{SY2}\bput[1pt](.6){\scriptsize $y$}
\ncline{->}{SX2}{SX3}\aput[1pt](.5){\scriptsize $x$}
\ncline{->}{SX2}{SX2Y}\bput[1pt](.5){\scriptsize $y$}
\ncline{->}{SXY}{SXYX}\aput[1pt](.5){\scriptsize $x$}
\ncline{->}{SXY}{SXY2}\bput[1pt](.5){\scriptsize $y$}
\ncline{->}{SX2Y}{SX2YX}\aput[1pt](.5){\scriptsize $x$}
\ncline{->}{SX2Y}{SX2Y2}\bput[1pt](.5){\scriptsize $y$}
\ncline{->}{SXYX}{SXYX2}\aput[1pt](.5){\scriptsize $x$}
\ncline{->}{SXYX}{SXYXY}\bput[1pt](.5){\scriptsize $y$}
\ncarc[linestyle=dashed, dash=0.7pt 2pt]{->}{SYX}{SE}
\ncarc[linestyle=dashed, dash=0.7pt 2pt,arcangle=12]{->}{SY2}{SX}
\ncarc[linestyle=dashed, dash=0.7pt 2pt,arcangle=26]{<-}{SXY}{SX3}
\ncarc[linestyle=dashed, dash=0.7pt 2pt,arcangle=22]{->}{SXY2}{SXYX}
\ncarc[linestyle=dashed, dash=0.7pt 2pt,arcangle=22]{<-}{SX2}{SX2YX}
\ncarc[linestyle=dashed, dash=0.7pt 2pt]{<-}{SXY}{SX2Y2}
\ncarc[linestyle=dashed, dash=0.7pt 2pt,arcangle=20]{->}{SXYX2}{SXYX}
\ncarc[linestyle=dashed, dash=0.7pt 2pt,arcangle=20]{<-}{SXYX}{SXYXY}
%%%%%%%%%%%
\ncarcbox[nodesep=.2cm,boxsize=.4,linearc=.4,arcangle=10,linestyle=dashed,linewidth=1.2pt,linecolor=gray]{SY}{SX}
\ncarcbox[nodesep=.2cm,boxsize=.4,linearc=.4,arcangle=10,linestyle=dashed,linewidth=1.2pt,linecolor=gray]{SXY}{SX2}
\ncarcbox[nodesep=.2cm,boxsize=.4,linearc=.4,arcangle=10,linestyle=dashed,linewidth=1.2pt,linecolor=gray]{SY2}{SYX}
\ncarcbox[nodesep=.2cm,boxsize=.4,linearc=.4,arcangle=10,linestyle=dashed,linewidth=1.2pt,linecolor=gray]{SX2Y}{SX3}
\ncarcbox[nodesep=.2cm,boxsize=.4,linearc=.4,arcangle=10,linestyle=dashed,linewidth=1.2pt,linecolor=gray]{SXY2}{SXYX}
\ncarcbox[nodesep=.2cm,boxsize=.4,linearc=.4,arcangle=10,linestyle=dashed,linewidth=1.2pt,linecolor=gray]{SX2Y2}{SX2YX}
\ncarcbox[nodesep=.2cm,boxsize=.4,linearc=.4,arcangle=10,linestyle=dashed,linewidth=1.2pt,linecolor=gray]{SXYXY}{SXYX2}
\nccurve[angleA=180,angleB=180,ncurv=1.4,linestyle=dashed, nodesep=5pt,linewidth=1.2pt,linecolor=gray]{->}{SX2YX}{SX2}
\nccurve[angleA=0,angleB=0,ncurvA=3,ncurvB=1,linestyle=dashed, nodesep=5pt,linewidth=1.2pt,linecolor=gray]{->}{SXYXY}{SXY2}
\nput[labelsep=17pt]{90}{SE}{\textcolor{gray}{\boldmath{$\sigma^c_{\!\varepsilon}$}}}
\nput[labelsep=26pt]{90}{SX}{\textcolor{gray}{\boldmath{$\sigma^c_{\!x}$}}}
\nput[labelsep=25pt]{90}{SY}{\textcolor{gray}{\boldmath{$\sigma^c_{\!y}$}}}
\nput[labelsep=26pt]{88}{SX2}{\textcolor{gray}{\boldmath{$\sigma^c_{\!x^2}$}}}
\nput[labelsep=26pt]{90}{SXY}{\textcolor{gray}{\boldmath{$\sigma^c_{\!xy}$}}}
\nput[labelsep=19pt]{90}{SX2Y}{\textcolor{gray}{\boldmath{$\sigma^c_{\!x^2y}$}}}
\nput[labelsep=21pt]{90}{SXYX}{\textcolor{gray}{\boldmath{$\sigma^c_{\!xyx}$}}}
%
%%%%%%%%%%%
\pnode(9,1){C2}
\SpecialCoor
\rput(C2){\cnode{3mm}{GSE}}
\rput(GSE){\scriptsize$\sigma^c_\varepsilon$}
\rput([angle=180,nodesep=5mm,offset=0pt]GSE){\pnode{GI}}
\ncline{->}{GI}{GSE}
\rput([angle=135,nodesep=15mm,offset=0pt]GSE){\cnode[doubleline=true]{3.5mm}{GSX}}
\rput(GSX){\scriptsize$\sigma^c_x$}
\rput([angle=45,nodesep=15mm,offset=0pt]GSE){\cnode[doubleline=true]{3.5mm}{GSY}}
\rput(GSY){\scriptsize$\sigma^c_y$}
\rput([angle=135,nodesep=15mm,offset=0pt]GSX){\cnode[doubleline=true]{3.5mm}{GSX2}}
\rput(GSX2){\scriptsize$\sigma^c_{\!x^2}$}
\rput([angle=45,nodesep=15mm,offset=0pt]GSX){\cnode[doubleline=true]{3.5mm}{GSXY}}
\rput(GSXY){\scriptsize$\sigma^c_{\!xy}$}
\ncline{->}{GSE}{GSX}\aput[1pt](.5){\scriptsize $x$}
\ncarc[arcangle=16]{<-}{GSY}{GSE}\aput[1pt](.5){\scriptsize $y$}
\ncline{->}{GSX}{GSXY}\bput[1pt](.5){\scriptsize $y$}
\ncarc[arcangle=16]{<-}{GSE}{GSY}\aput[1pt](.5){\scriptsize $x$}
\ncarc{<-}{GSX}{GSY}\aput[1pt](.5){\scriptsize $y$}
\ncarc[arcangle=16]{->}{GSX2}{GSX}\aput[1pt](.5){\scriptsize $y$}
\ncarc[arcangle=16]{->}{GSX}{GSX2}\aput[1pt](.5){\scriptsize $x$}
\ncarc{->}{GSX2}{GSXY}\aput[1pt](.5){\scriptsize $x$}
\nccurve[angleA=135,angleB=45,ncurv=4]{->}{GSXY}{GSXY}\aput[0.5pt](.50){\scriptsize $x,y$}
\NormalCoor
\end{pspicture}
%%%%
\end{center}
\caption{\scriptsize Construction of the transition graph of the children automaton ${\cal A}_N^c$ from  the transition tree of the Nerode automaton ${\cal A}_N$ from Fig. \ref{fig:N.ri.wri.1}.}\label{fig:N.ri.wri.4}
%%%%
\end{figure}
%%%%

According to Theorem \ref{th:hom.im}, the number of states of the automaton ${\cal A}_{\varphi^{\textrm{ri}}}$
is less than or equal to the~number of states of the Nerode automaton ${\cal
A}_N$, for every fuzzy automaton $\cal A$.~This example shows that~${\cal A}_{\varphi^{\textrm{ri}}}$~can be strictly~smaller~than ${\cal
A}_N$.~The example also shows that the greatest weakly right invariant
fuzzy~quasi-order $\varphi^{\textrm{wri}}$ can give better results in determinization
than the greatest right invariant
fuzzy quasi-order~$\varphi^{\textrm{ri}}$.

Finally, we find that all aftersets of fuzzy quasi-orders $\varphi^{\textrm{ri}}$
and  $\varphi^{\textrm{wri}}$ are different, which means that~$\varphi^{\textrm{ri}}$
and  $\varphi^{\textrm{wri}}$ do not reduce the number of states of the automaton
$\cal A$, but despite this, they produce~crisp-deterministic fuzzy automata which are smaller than the Nerode automaton of $\cal A$.
\end{example}

\begin{example}\label{ex:Nc.ri}\rm
Let ${\cal A}$ be a Boolean automaton over the two-element alphabet $X=\{x,y\}$ given by the transition graph shown in Fig.~\ref{fig:Nc.ri} a).~The transition
graphs of the Nerode automaton ${\cal A}_N$ and its children auto\-maton ${\cal
A}^c_N$, constructed by means of Algorithms \ref{alg:A.phi} and \ref{alg:A.phi.c},
are represented by Fig. \ref{fig:Nc.ri}~b)~and~c).~Clearly, the Nerode automaton ${\cal A}_N$ has 7 states, and its children auto\-maton ${\cal
A}^c_N$ has 6 states.

On the other hand, both the greatest right invariant fuzzy~quasi-order  $\varphi^{\textrm{ri}}$
and the greatest weakly right
invariant fuzzy~quasi-order  $\varphi^{\textrm{wri}}$ on $\cal A$, computed
using Algorithms \ref{alg:gri} and \ref{alg:gwri},~are~equal~to~the~equality relation on the set of states of $\cal A$, so both automata  ${\cal A}_{\varphi^{\textrm{ri}}}$
and  ${\cal A}_{\varphi^{\textrm{wri}}}$ are isomorphic to ${\cal A}_N$.~Therefore,
we have that  $|{\cal A}^c_N|<|{\cal A}_N|=|{\cal A}_{\varphi^{\textrm{ri}}}|=|{\cal A}_{\varphi^{\textrm{wri}}}|$,~which demonstrates that construction of the~children~auto\-maton, applied to the Nerode automaton of $\cal A$, can~give better
results in determinization of $\cal A$~than construc\-tions based on the greatest right invariant and weakly right invariant fuzzy~quasi-orders on $\cal A$.

\begin{figure}
%%%%
\begin{center}
%%%%
\psset{unit=1cm}
\newpsobject{showgrid}{psgrid}{subgriddiv=1,griddots=10,gridlabels=6pt}
%%%
\begin{pspicture}(-6,-0.2)(11,4.8)%\showgrid
\rput(-6,4.8){\textsf{a)}}
\rput(-1.1,4.8){\textsf{b)}}
\rput(4.9,4.8){\textsf{c)}}
\pnode(-5,2){AP1}
\SpecialCoor
\rput(AP1){\cnode[doubleline=true]{3mm}{AP1A1}}
\rput(AP1A1){\scriptsize$a_1$}
\rput([angle=-30,nodesep=25mm,offset=0pt]AP1A1){\cnode{2.7mm}{AP1A2}}
\rput(AP1A2){\scriptsize$a_2$}
\rput([angle=30,nodesep=25mm,offset=0pt]AP1A1){\cnode[doubleline=true]{3mm}{AP1A3}}
\rput(AP1A3){\scriptsize$a_3$}
\rput([angle=180,nodesep=5mm,offset=0pt]AP1A1){\pnode{AP1I}}
\ncline{->}{AP1I}{AP1A1}
\ncarc[arcangle=16]{<-}{AP1A1}{AP1A2}\aput[1pt](.5){\scriptsize $x,y$}
\ncarc[arcangle=16]{<-}{AP1A2}{AP1A1}\aput[1pt](.5){\scriptsize $x$}
\ncarc[arcangle=16]{<-}{AP1A1}{AP1A3}\aput[1pt](.5){\scriptsize $x$}
\ncarc[arcangle=16]{<-}{AP1A3}{AP1A1}\aput[1pt](.5){\scriptsize $y$}
\ncarc[arcangle=16]{<-}{AP1A2}{AP1A3}\aput[1pt](.5){\scriptsize $y$}
\ncarc[arcangle=16]{<-}{AP1A3}{AP1A2}\aput[1pt](.5){\scriptsize $x$}
%
%\nccircle[angleA=-90]{<-}{AP1A2}{0.4}\bput[0pt](.40){\scriptsize $y$}
\nccurve[angleA=45,angleB=-45,ncurv=4]{->}{AP1A2}{AP1A2}\aput[0.5pt](.50){\scriptsize $y$}
\NormalCoor
%%%%%%%%%%%
\pnode(2.5,0){C2}
\SpecialCoor
\rput(C2){\cnode[doubleline=true]{3.5mm}{GSE}}
\rput(GSE){\scriptsize$\sigma_\varepsilon$}
\rput([angle=180,nodesep=5mm,offset=0pt]GSE){\pnode{GI}}
\ncline{->}{GI}{GSE}
\rput([angle=135,nodesep=15mm,offset=0pt]GSE){\cnode{3mm}{GSX}}
\rput(GSX){\scriptsize$\sigma_x$}
\rput([angle=45,nodesep=15mm,offset=0pt]GSE){\cnode[doubleline=true]{3.5mm}{GSY}}
\rput(GSY){\scriptsize$\sigma_y$}
\rput([angle=135,nodesep=15mm,offset=0pt]GSX){\cnode[doubleline=true]{3.5mm}{GSX2}}
\rput(GSX2){\scriptsize$\sigma_{\!x^2}$}
\rput([angle=45,nodesep=15mm,offset=0pt]GSX){\cnode[doubleline=true]{3.5mm}{GSXY}}
\rput(GSXY){\scriptsize$\sigma_{\!xy}$}
\rput([angle=90,nodesep=13mm,offset=0pt]GSX2){\cnode[doubleline=true]{3.5mm}{GSX2Y}}
\rput(GSX2Y){\scriptsize$\sigma_{\!x^2y}$}
\rput([angle=90,nodesep=13mm,offset=0pt]GSXY){\cnode[doubleline=true]{3.5mm}{GSXYX}}
\rput(GSXYX){\scriptsize$\sigma_{\!xyx}$}
\ncline{->}{GSE}{GSX}\aput[1pt](.5){\scriptsize $x$}
\ncarc[arcangle=16]{<-}{GSY}{GSE}\aput[1pt](.5){\scriptsize $y$}
\ncline{->}{GSX}{GSX2}\aput[1pt](.5){\scriptsize $x$}
\ncline{->}{GSX}{GSXY}\bput[1pt](.5){\scriptsize $y$}
\ncline{->}{GSX2}{GSXY}\bput[1pt](.5){\scriptsize $x$}
\ncarc[arcangle=16]{<-}{GSE}{GSY}\aput[1pt](.5){\scriptsize $x$}
\ncarc{<-}{GSX}{GSY}\aput[1pt](.5){\scriptsize $y$}
\ncarc[arcangle=16]{->}{GSX2}{GSX2Y}\aput[1pt](.5){\scriptsize $y$}
\ncarc[arcangle=16]{->}{GSX2Y}{GSX2}\aput[1pt](.5){\scriptsize $x$}
\ncline{->}{GSXY}{GSXYX}\bput[1pt](.5){\scriptsize $x,y$}
\ncline{->}{GSX2Y}{GSXY}\bput[1pt](.5){\scriptsize $y$}
\nccurve[angleA=135,angleB=45,ncurv=4]{->}{GSXYX}{GSXYX}\aput[0.5pt](.50){\scriptsize $x,y$}
%\nccircle[angleA=90]{->}{GSX2}{0.4}\bput[0pt](.40){\scriptsize $x$}
%\nccircle[angleB=180]{<-}{GSXYX}{0.4}\bput[0pt](.50){\scriptsize $x,y$}
%
\NormalCoor
%%%%%%%%%%%
\pnode(8.5,0){C2}
\SpecialCoor
\rput(C2){\cnode[doubleline=true]{3.6mm}{GSE}}
\rput(GSE){\scriptsize$\sigma^c_\varepsilon$}
\rput([angle=180,nodesep=5mm,offset=0pt]GSE){\pnode{GI}}
\ncline{->}{GI}{GSE}
\rput([angle=135,nodesep=15mm,offset=0pt]GSE){\cnode{3mm}{GSX}}
\rput(GSX){\scriptsize$\sigma^c_x$}
\rput([angle=45,nodesep=15mm,offset=0pt]GSE){\cnode[doubleline=true]{3.6mm}{GSY}}
\rput(GSY){\scriptsize$\sigma^c_y$}
\rput([angle=135,nodesep=15mm,offset=0pt]GSX){\cnode[doubleline=true]{3.6mm}{GSX2}}
\rput(GSX2){\scriptsize$\sigma^c_{\!x^2}$}
\rput([angle=45,nodesep=15mm,offset=0pt]GSX){\cnode[doubleline=true]{3.6mm}{GSXY}}
\rput(GSXY){\scriptsize$\sigma^c_{\!xy}$}
\rput([angle=90,nodesep=13mm,offset=0pt]GSX2){\cnode[doubleline=true]{3.6mm}{GSX2Y}}
\rput(GSX2Y){\scriptsize$\sigma^c_{\!x^2y}$}
\ncline{->}{GSE}{GSX}\aput[1pt](.5){\scriptsize $x$}
\ncarc[arcangle=16]{<-}{GSY}{GSE}\aput[1pt](.5){\scriptsize $y$}
\ncline{->}{GSX}{GSX2}\aput[1pt](.5){\scriptsize $x$}
\ncline{->}{GSX}{GSXY}\bput[1pt](.5){\scriptsize $y$}
\ncline{->}{GSX2}{GSXY}\bput[1pt](.5){\scriptsize $x$}
\ncarc[arcangle=16]{<-}{GSE}{GSY}\aput[1pt](.5){\scriptsize $x$}
\ncarc{<-}{GSX}{GSY}\aput[1pt](.5){\scriptsize $y$}
\ncarc[arcangle=16]{->}{GSX2}{GSX2Y}\aput[1pt](.5){\scriptsize $y$}
\ncarc[arcangle=16]{->}{GSX2Y}{GSX2}\aput[1pt](.5){\scriptsize $x$}
\ncline{->}{GSX2Y}{GSXY}\bput[1pt](.5){\scriptsize $y$}
\nccurve[angleA=45,angleB=-45,ncurv=4]{->}{GSXY}{GSXY}\aput[0.5pt](.50){\scriptsize $x,y$}
%\nccircle[angleA=90]{->}{GSX2}{0.4}\bput[0pt](.40){\scriptsize $x$}
%\nccircle[angleB=180]{<-}{GSXYX}{0.4}\bput[0pt](.50){\scriptsize $x,y$}
%
\NormalCoor
\end{pspicture}
%%%%
\end{center}
\caption{\scriptsize Transition graphs of the fuzzy automaton ${\cal A}$ from Example \ref{ex:Nc.ri} (a)), its Nerode automaton ${\cal A}_N$ (b)), and the children automaton ${\cal A}^c_N$ (c)) .}\label{fig:Nc.ri}
%%%%
\end{figure}
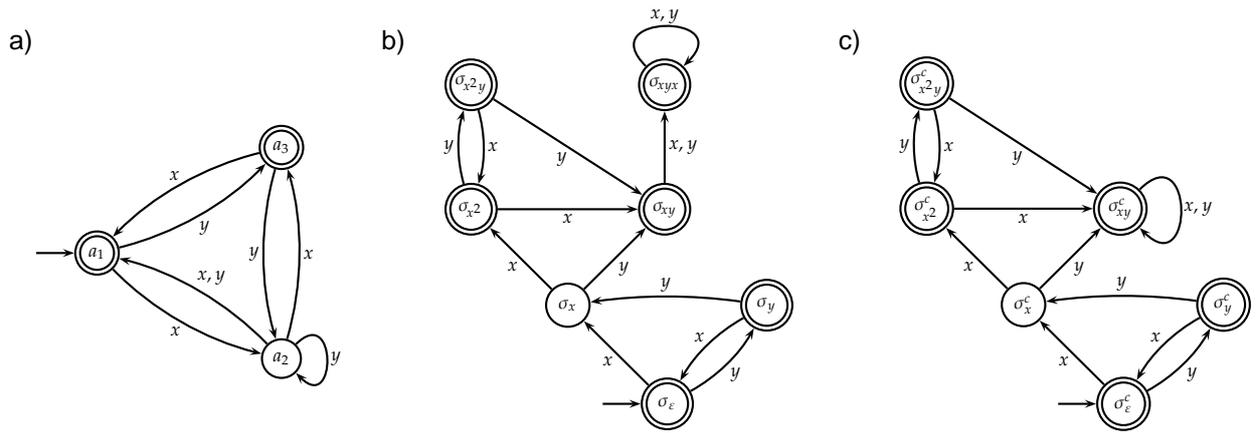
%%%%
\end{example}

\begin{example}\label{ex:ri.Nc}\rm
Let ${\cal A}$ be a Boolean automaton over the two-element alphabet $X=\{x,y\}$ given by the transition graph shown in Fig.~\ref{fig:ri.Nc} a).~The Nerode
automaton ${\cal A}_N$ and the children automaton ${\cal A}^c_N$~are~repre\-sented
by graphs in Fig.~\ref{fig:ri.Nc} b) and c).~Moreover, using Algorithms \ref{alg:gri} and \ref{alg:gwri} we obtain that
\[
\varphi^{\textrm{ri}}=\varphi^{\textrm{wri}}=\begin{bmatrix} 1 & 0 & 0 \\ 1 & 1 & 1 \\ 0 & 0 &
1\end{bmatrix},
\]
so automata   ${\cal A}_{\varphi^{\textrm{ri}}}$ and ${\cal A}_{\varphi^{\textrm{wri}}}$
are  mutually isomorphic and are given by the transition graph in Fig.~\ref{fig:ri.Nc}~d),~and we conclude that $|{\cal A}_{\varphi^{\textrm{ri}}}|=|{\cal A}_{\varphi^{\textrm{wri}}}|<|{\cal A}^c_N|<|{\cal A}_N|$.

Thus, in contrast to the previous one, this example shows that there are cases where~determinization of~a fuzzy automaton $\cal A$ by means of the greatest right invariant and weakly right invariant~fuzzy~quasi-orders can give better results than construction of the children automaton of the~Nerode~automaton~of~$\cal A$.
\begin{figure}
%%%%
\begin{center}
%%%%
\psset{unit=1cm}
\newpsobject{showgrid}{psgrid}{subgriddiv=1,griddots=10,gridlabels=6pt}
%%%
\begin{pspicture}(-6.5,0)(11,4.8)%\showgrid
\rput(-6.5,4.8){\textsf{a)}}
\rput(-2.5,4.8){\textsf{b)}}
\rput(2.2,4.3){\textsf{c)}}
\rput(6.9,4.3){\textsf{d)}}
\pnode(-6.4,2.5){AP1}
\SpecialCoor
\rput(AP1){\cnode[doubleline=true]{3mm}{AP1A1}}
\rput(AP1A1){\scriptsize$a_1$}
\rput([angle=-30,nodesep=25mm,offset=0pt]AP1A1){\cnode[doubleline=true]{3mm}{AP1A2}}
\rput(AP1A2){\scriptsize$a_2$}
\rput([angle=30,nodesep=25mm,offset=0pt]AP1A1){\cnode{2.7mm}{AP1A3}}
\rput(AP1A3){\scriptsize$a_3$}
\rput([angle=270,nodesep=5mm,offset=0pt]AP1A1){\pnode{AP1I}}
\ncline{->}{AP1I}{AP1A1}
\ncline{<-}{AP1A1}{AP1A2}\bput[1pt](.5){\scriptsize $x,y$}
\ncline{<-}{AP1A3}{AP1A1}\bput[1pt](.5){\scriptsize $x,y$}
\ncline{->}{AP1A3}{AP1A2}\aput[1pt](.5){\scriptsize $x$}
%
%\nccircle[angleA=-90]{<-}{AP1A2}{0.4}\bput[0pt](.40){\scriptsize $y$}
\nccurve[angleA=0,angleB=-90,ncurv=4]{->}{AP1A2}{AP1A2}\aput[0.5pt](.50){\scriptsize $x,y$}
\nccurve[angleA=90,angleB=0,ncurv=4]{->}{AP1A3}{AP1A3}\aput[0.5pt](.50){\scriptsize $y$}
\nccurve[angleA=180,angleB=90,ncurv=4]{->}{AP1A1}{AP1A1}\aput[0.5pt](.50){\scriptsize $x$}
\NormalCoor
%%%%%%%%%%%
\pnode(-0.6,0){C2}
\SpecialCoor
\rput(C2){\cnode[doubleline=true]{3.5mm}{GSE}}
\rput(GSE){\scriptsize$\sigma_\varepsilon$}
\rput([angle=180,nodesep=5mm,offset=0pt]GSE){\pnode{GI}}
\ncline{->}{GI}{GSE}
\rput([angle=135,nodesep=15mm,offset=0pt]GSE){\cnode[doubleline=true]{3.5mm}{GSX}}
\rput(GSX){\scriptsize$\sigma_x$}
\rput([angle=45,nodesep=15mm,offset=0pt]GSE){\cnode{3mm}{GSY}}
\rput(GSY){\scriptsize$\sigma_y$}
\rput([angle=90,nodesep=13mm,offset=0pt]GSX){\cnode[doubleline=true]{3.5mm}{GSX2}}
\rput(GSX2){\scriptsize$\sigma_{\!x^2}$}
\rput([angle=90,nodesep=13mm,offset=0pt]GSY){\cnode[doubleline=true]{3.5mm}{GSYX}}
\rput(GSYX){\scriptsize$\sigma_{\!yx}$}
\rput([angle=90,nodesep=13mm,offset=0pt]GSYX){\cnode[doubleline=true]{3.5mm}{GSYX2}}
\rput(GSYX2){\scriptsize$\sigma_{\!yx^2}$}
\ncline{->}{GSE}{GSX}\aput[1pt](.5){\scriptsize $x$}
\ncline{<-}{GSY}{GSE}\aput[1pt](.5){\scriptsize $y$}
\ncline{->}{GSX}{GSX2}\aput[1pt](.5){\scriptsize $x$}
\ncline{->}{GSX}{GSY}\aput[1pt](.5){\scriptsize $y$}
\ncline{->}{GSY}{GSYX}\bput[1pt](.5){\scriptsize $x$}
\ncline{->}{GSYX2}{GSX2}\bput[1pt](.5){\scriptsize $x,y$}
\ncline{->}{GSYX}{GSYX2}\bput[1pt](.5){\scriptsize $x,y$}
\nccurve[angleA=180,angleB=90,ncurv=4]{->}{GSX2}{GSX2}\aput[0.5pt](.70){\scriptsize $x,y$}
\nccurve[angleA=45,angleB=-45,ncurv=4]{->}{GSY}{GSY}\aput[0.5pt](.50){\scriptsize $y$}
\NormalCoor
%%%%%%%%%%%
\pnode(4.2,0){C2}
\SpecialCoor
\rput(C2){\cnode[doubleline=true]{3.5mm}{GSE}}
\rput(GSE){\scriptsize$\sigma^c_\varepsilon$}
\rput([angle=180,nodesep=5mm,offset=0pt]GSE){\pnode{GI}}
\ncline{->}{GI}{GSE}
\rput([angle=135,nodesep=15mm,offset=0pt]GSE){\cnode[doubleline=true]{3.5mm}{GSX}}
\rput(GSX){\scriptsize$\sigma^c_x$}
\rput([angle=45,nodesep=15mm,offset=0pt]GSE){\cnode{3mm}{GSY}}
\rput(GSY){\scriptsize$\sigma^c_y$}
\rput([angle=90,nodesep=13mm,offset=0pt]GSX){\cnode[doubleline=true]{3.5mm}{GSX2}}
\rput(GSX2){\scriptsize$\sigma^c_{\!x^2}$}
\rput([angle=90,nodesep=13.5mm,offset=0pt]GSY){\cnode[doubleline=true]{3.5mm}{GSYX}}
\rput(GSYX){\scriptsize$\sigma^c_{\!yx}$}
\ncline{->}{GSE}{GSX}\aput[1pt](.5){\scriptsize $x$}
\ncline{<-}{GSY}{GSE}\aput[1pt](.5){\scriptsize $y$}
\ncline{->}{GSX}{GSX2}\aput[1pt](.5){\scriptsize $x$}
\ncline{->}{GSX}{GSY}\aput[1pt](.5){\scriptsize $y$}
\ncline{->}{GSY}{GSYX}\bput[1pt](.5){\scriptsize $x$}
\ncline{->}{GSYX}{GSX2}\bput[1pt](.5){\scriptsize $x,y$}
\nccurve[angleA=135,angleB=45,ncurv=4]{->}{GSX2}{GSX2}\aput[0.5pt](.50){\scriptsize $x,y$}
\nccurve[angleA=45,angleB=-45,ncurv=4]{->}{GSY}{GSY}\aput[0.5pt](.50){\scriptsize $y$}
\NormalCoor
%%%%%%%%%%%
\pnode(8.8,0){C2}
\SpecialCoor
\rput(C2){\cnode[doubleline=true]{3.5mm}{GSE}}
\rput(GSE){\scriptsize$\varphi^{\textrm{ri}}_\varepsilon$}
\rput([angle=180,nodesep=5mm,offset=0pt]GSE){\pnode{GI}}
\ncline{->}{GI}{GSE}
\rput([angle=135,nodesep=15mm,offset=0pt]GSE){\cnode[doubleline=true]{3.5mm}{GSX}}
\rput(GSX){\scriptsize$\varphi^{\textrm{ri}}_x$}
\rput([angle=45,nodesep=15mm,offset=0pt]GSE){\cnode{3mm}{GSY}}
\rput(GSY){\scriptsize$\varphi^{\textrm{ri}}_y$}
\rput([angle=45,nodesep=15mm,offset=0pt]GSX){\cnode[doubleline=true]{3.5mm}{GSX2}}
\rput(GSX2){\scriptsize$\varphi^{\textrm{ri}}_{\!x^2}$}
\ncline{->}{GSE}{GSX}\aput[1pt](.5){\scriptsize $x$}
\ncline{<-}{GSY}{GSE}\aput[1pt](.5){\scriptsize $y$}
\ncline{->}{GSX}{GSX2}\aput[1pt](.5){\scriptsize $x$}
\ncline{->}{GSX}{GSY}\aput[1pt](.5){\scriptsize $y$}
\ncline{->}{GSY}{GSX2}\bput[1pt](.5){\scriptsize $x$}
\nccurve[angleA=135,angleB=45,ncurv=4]{->}{GSX2}{GSX2}\aput[0.5pt](.50){\scriptsize $x,y$}
\nccurve[angleA=45,angleB=-45,ncurv=4]{->}{GSY}{GSY}\aput[0.5pt](.50){\scriptsize $y$}
\NormalCoor
\end{pspicture}
%%%%
\end{center}
\caption{\scriptsize Transition graphs of the fuzzy automaton ${\cal A}$ from Example \ref{ex:ri.Nc} (a)), its Nerode automaton ${\cal A}_N$ (b)),  the children automaton ${\cal A}^c_N$ of ${\cal A}_N$ (c)), and the automaton
${\cal A}_{\varphi^{\textrm{ri}}}\cong {\cal A}_{\varphi^{\textrm{wri}}}$ (d)).}\label{fig:ri.Nc}
%%%%
\end{figure}
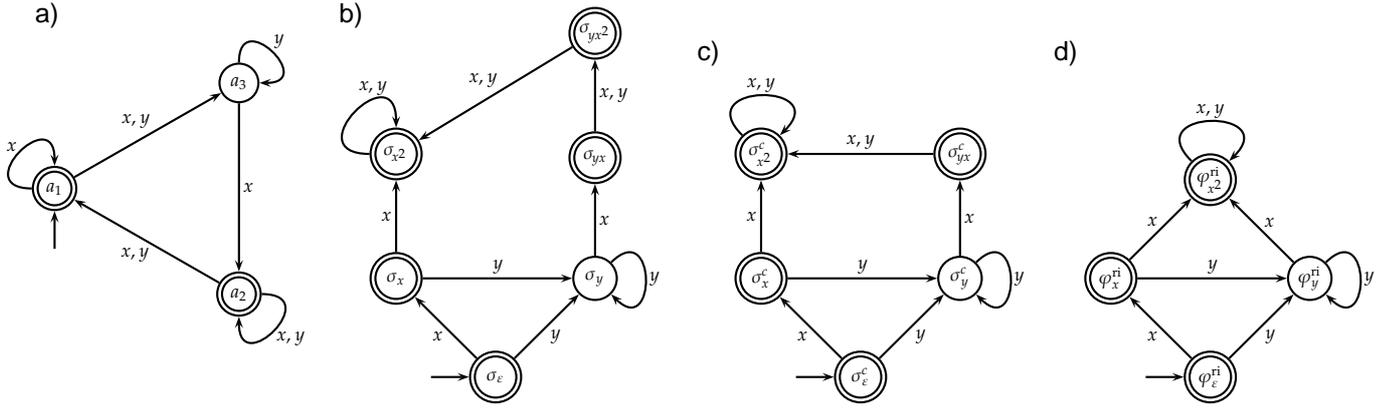
%%%%
\end{example}

\begin{example}\label{ex:all.AN}\rm
Let ${\cal A}$ be a Boolean automaton over the two-element alphabet $X=\{x,y\}$ given by the transition graph shown in Fig.~\ref{fig:All.N} a).~When we compute    $\varphi^{\textrm{ri}}$
and   $\varphi^{\textrm{wri}}$  we obtain that both of them~are~equal to
the equality relation on the set of states of $\cal A$, and therefore, both automata  ${\cal A}_{\varphi^{\textrm{ri}}}$
and  ${\cal A}_{\varphi^{\textrm{wri}}}$~are isomorphic to the Nerode automaton ${\cal A}_N$, which is represented by the transition graph in Fig.~\ref{fig:All.N} b).~More\-over, we obtain that the children automaton ${\cal A}^c_N$ of ${\cal A}_N$ is also isomorphic to the  Nerode automaton ${\cal A}_N$.~Hence, in
this case none of the methods discussed in this paper does not give an automaton with smaller number of states than the Nerode automaton ${\cal A}_N$.~It should be noted that the Nerode automaton ${\cal A}_N$ is not minimal, the minimal deterministic automaton equivalent to ${\cal A}_N$ is represented by the graph shown in Fig.~\ref{fig:All.N} c).
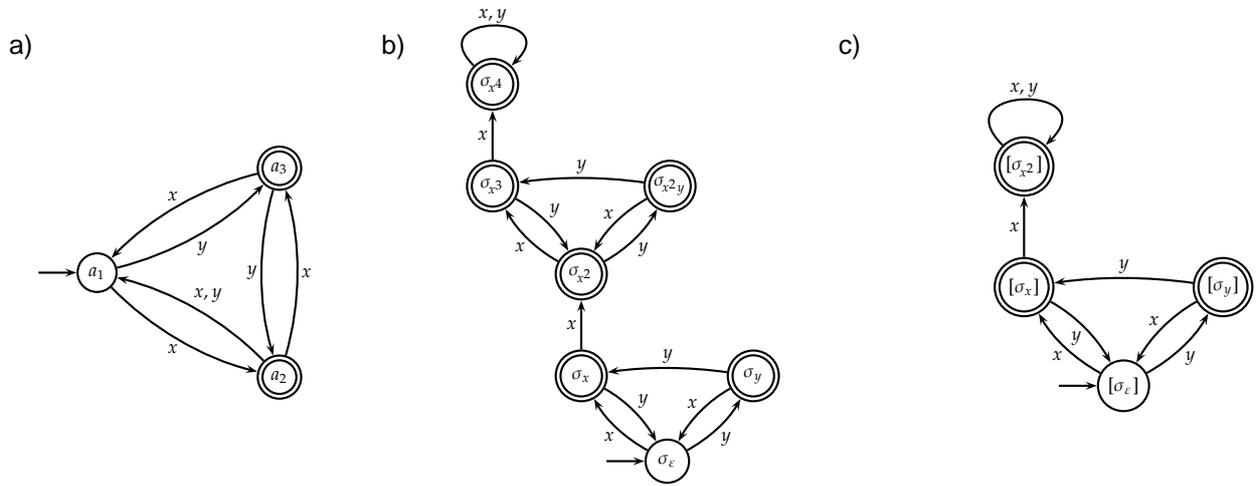
\begin{figure}
%%%%
\begin{center}
%%%%
\psset{unit=1cm}
\newpsobject{showgrid}{psgrid}{subgriddiv=1,griddots=10,gridlabels=6pt}
%%%
\begin{pspicture}(-6,-0.1)(11,5.7)%\showgrid
\rput(-6,5.5){\textsf{a)}}
\rput(-1.1,5.5){\textsf{b)}}
\rput(4.9,5.5){\textsf{c)}}
\pnode(-5,2.5){AP1}
\SpecialCoor
\rput(AP1){\cnode{2.7mm}{AP1A1}}
\rput(AP1A1){\scriptsize$a_1$}
\rput([angle=-30,nodesep=25mm,offset=0pt]AP1A1){\cnode[doubleline=true]{3mm}{AP1A2}}
\rput(AP1A2){\scriptsize$a_2$}
\rput([angle=30,nodesep=25mm,offset=0pt]AP1A1){\cnode[doubleline=true]{3mm}{AP1A3}}
\rput(AP1A3){\scriptsize$a_3$}
\rput([angle=180,nodesep=5mm,offset=0pt]AP1A1){\pnode{AP1I}}
\ncline{->}{AP1I}{AP1A1}
\ncarc[arcangle=16]{<-}{AP1A1}{AP1A2}\aput[1pt](.5){\scriptsize $x,y$}
\ncarc[arcangle=16]{<-}{AP1A2}{AP1A1}\aput[1pt](.5){\scriptsize $x$}
\ncarc[arcangle=16]{<-}{AP1A1}{AP1A3}\aput[1pt](.5){\scriptsize $x$}
\ncarc[arcangle=16]{<-}{AP1A3}{AP1A1}\aput[1pt](.5){\scriptsize $y$}
\ncarc[arcangle=16]{<-}{AP1A2}{AP1A3}\aput[1pt](.5){\scriptsize $y$}
\ncarc[arcangle=16]{<-}{AP1A3}{AP1A2}\aput[1pt](.5){\scriptsize $x$}
\NormalCoor
%%%%%%%%%%%
\pnode(2.5,0){C2}
\SpecialCoor
\rput(C2){\cnode{3mm}{GSE}}
\rput(GSE){\scriptsize$\sigma_\varepsilon$}
\rput([angle=180,nodesep=5mm,offset=0pt]GSE){\pnode{GI}}
\ncline{->}{GI}{GSE}
\rput([angle=135,nodesep=13mm,offset=0pt]GSE){\cnode[doubleline=true]{3.5mm}{GSX}}
\rput(GSX){\scriptsize$\sigma_x$}
\rput([angle=45,nodesep=13mm,offset=0pt]GSE){\cnode[doubleline=true]{3.5mm}{GSY}}
\rput(GSY){\scriptsize$\sigma_y$}
\rput([angle=90,nodesep=10mm,offset=0pt]GSX){\cnode[doubleline=true]{3.5mm}{GSX2}}
\rput(GSX2){\scriptsize$\sigma_{\!x^2}$}
\rput([angle=45,nodesep=13mm,offset=0pt]GSX2){\cnode[doubleline=true]{3.5mm}{GSX2Y}}
\rput(GSX2Y){\scriptsize$\sigma_{\!x^2y}$}
\rput([angle=135,nodesep=13mm,offset=0pt]GSX2){\cnode[doubleline=true]{3.5mm}{GSX3}}
\rput(GSX3){\scriptsize$\sigma_{\!x^3}$}
\rput([angle=90,nodesep=10mm,offset=0pt]GSX3){\cnode[doubleline=true]{3.5mm}{GSX4}}
\rput(GSX4){\scriptsize$\sigma_{\!x^4}$}
\ncarc[arcangle=16]{->}{GSE}{GSX}\aput[1pt](.5){\scriptsize $x$}
\ncarc[arcangle=16]{<-}{GSY}{GSE}\aput[1pt](.5){\scriptsize $y$}
\ncarc[arcangle=16]{->}{GSX}{GSE}\aput[1pt](.5){\scriptsize $y$}
\ncline{->}{GSX}{GSX2}\aput[1pt](.5){\scriptsize $x$}
\ncarc[arcangle=16]{<-}{GSE}{GSY}\aput[1pt](.5){\scriptsize $x$}
\ncarc{<-}{GSX}{GSY}\aput[1pt](.5){\scriptsize $y$}
\ncarc[arcangle=16]{->}{GSX2}{GSX3}\aput[1pt](.5){\scriptsize $x$}
\ncarc[arcangle=16]{->}{GSX3}{GSX2}\aput[1pt](.5){\scriptsize $y$}
\ncarc[arcangle=16]{<-}{GSX2}{GSX2Y}\aput[1pt](.5){\scriptsize $x$}
\ncarc[arcangle=16]{<-}{GSX2Y}{GSX2}\aput[1pt](.5){\scriptsize $y$}
\ncarc{<-}{GSX3}{GSX2Y}\aput[1pt](.5){\scriptsize $y$}
\ncline{->}{GSX3}{GSX4}\aput[1pt](.5){\scriptsize $x$}
\nccurve[angleA=135,angleB=45,ncurv=4]{->}{GSX4}{GSX4}\aput[0.5pt](.50){\scriptsize $x,y$}
%\nccircle[angleA=90]{->}{GSX2}{0.4}\bput[0pt](.40){\scriptsize $x$}
%\nccircle[angleB=180]{<-}{GSXYX}{0.4}\bput[0pt](.50){\scriptsize $x,y$}
%
\NormalCoor
%%%%%%%%%%%
\pnode(8.5,1){C2}
\SpecialCoor
\rput(C2){\cnode{3.5mm}{GSE}}
\rput(GSE){\scriptsize$[\sigma_\varepsilon]$}
\rput([angle=180,nodesep=5mm,offset=0pt]GSE){\pnode{GI}}
\ncline{->}{GI}{GSE}
\rput([angle=135,nodesep=15mm,offset=0pt]GSE){\cnode[doubleline=true]{4mm}{GSX}}
\rput(GSX){\scriptsize$[\sigma_x]$}
\rput([angle=45,nodesep=15mm,offset=0pt]GSE){\cnode[doubleline=true]{4mm}{GSY}}
\rput(GSY){\scriptsize$[\sigma_y]$}
\rput([angle=90,nodesep=12mm,offset=0pt]GSX){\cnode[doubleline=true]{4mm}{GSX2}}
\rput(GSX2){\scriptsize$[\sigma_{\!x^2}]$}
\ncarc[arcangle=16]{<-}{GSY}{GSE}\aput[1pt](.5){\scriptsize $y$}
\ncline{->}{GSX}{GSX2}\aput[1pt](.5){\scriptsize $x$}
\ncarc[arcangle=16]{->}{GSX}{GSE}\bput[1pt](.5){\scriptsize $y$}
\ncarc[arcangle=16]{<-}{GSE}{GSY}\aput[1pt](.5){\scriptsize $x$}
\ncarc{<-}{GSX}{GSY}\aput[1pt](.5){\scriptsize $y$}
\ncarc[arcangle=16]{->}{GSE}{GSX}\aput[1pt](.5){\scriptsize $x$}
\nccurve[angleA=135,angleB=45,ncurv=4]{->}{GSX2}{GSX2}\aput[0.5pt](.50){\scriptsize $x,y$}
%\nccircle[angleA=90]{->}{GSX2}{0.4}\bput[0pt](.40){\scriptsize $x$}
%\nccircle[angleB=180]{<-}{GSXYX}{0.4}\bput[0pt](.50){\scriptsize $x,y$}
%
\NormalCoor
\end{pspicture}
%%%%
\end{center}
\caption{\scriptsize Transition graphs of the fuzzy automaton ${\cal A}$ from Example \ref{ex:all.AN} (a)), its Nerode automaton ${\cal A}_N$ (b)), and the minimal crisp-deterministic fuzzy automaton equivalent to ${\cal A}_N$ (c)) .}\label{fig:All.N}
%%%%
\end{figure}
%%%%
\end{example}

\begin{example}\label{ex:red.no.det}\rm
Let ${\cal A}$ be a Boolean automaton over the two-element alphabet $X=\{x,y\}$ given by the transition graph shown in Fig.~\ref{fig:red.no.det} a).~The
transition graph of the Nerode automaton ${\cal A}_N$ of $\cal A$ is given
in Fig.~\ref{fig:red.no.det}~b).~Using Algorithms \ref{alg:gri} and \ref{alg:gwri} we obtain that
\[
\varphi^{\textrm{ri}}=\varphi^{\textrm{wri}}=\begin{bmatrix}
1 & 1 & 0 & 0 & 0 \\
1 & 1 & 0 & 0 & 0 \\
0 & 0 & 1 & 0 & 0 \\
0 & 0 & 0 & 1 & 0 \\
0 & 0 & 1 & 0 & 1
\end{bmatrix},
\]
and we can easily show that ${\cal A}_{\varphi^{\textrm{ri}}}$ is isomorphic
to the Nerode automaton ${\cal A}_N$.~Note that $\varphi^{\textrm{ri}}$ has
4 different aftersets, which means that the afterset fuzzy automaton of $\cal A$ with respect to $\varphi^{\textrm{ri}}$ has 4 states. Therefore, although
$\varphi^{\textrm{ri}}$ reduces the number of states of $\cal A$, it does not give an automaton with smaller number of states than the Nerode automaton ${\cal A}_N$.
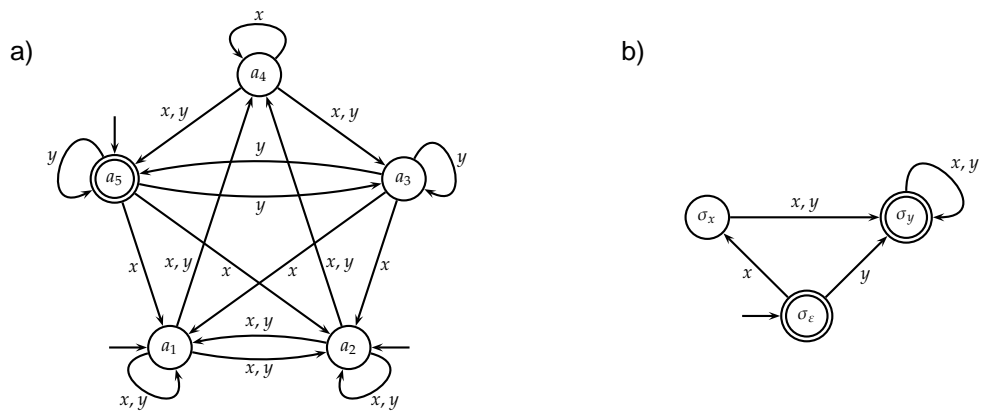
\begin{figure}
%%%%
\begin{center}
%%%%
\psset{unit=1.2cm}
\newpsobject{showgrid}{psgrid}{subgriddiv=1,griddots=10,gridlabels=6pt}%\showgrid
%%%
\begin{pspicture}(-3,-1.9)(8,2)%\showgrid
\rput(-2.6,1.9){\textsf{a)}}
\rput(4.1,1.9){\textsf{b)}}
\pnode(0,0){C}
\SpecialCoor
\degrees[100]
\rput([angle=65,nodesep=20mm,offset=0pt]C){\cnode{3mm}{A1}}
\rput(A1){\scriptsize$a_1$}
\rput([angle=85,nodesep=20mm,offset=0pt]C){\cnode{3mm}{A2}}
\rput(A2){\scriptsize$a_2$}
\rput([angle=5,nodesep=20mm,offset=0pt]C){\cnode{3mm}{A3}}
\rput(A3){\scriptsize$a_3$}
\rput([angle=25,nodesep=20mm,offset=0pt]C){\cnode{3mm}{A4}}
\rput(A4){\scriptsize$a_4$}
\rput([angle=45,nodesep=20mm,offset=0pt]C){\cnode[doubleline=true]{3.3mm}{A5}}
\rput(A5){\scriptsize$a_5$}
\degrees[360]
\rput([angle=180,nodesep=5mm,offset=0pt]A1){\pnode{I1}}
\ncline{->}{I1}{A1}
\rput([angle=0,nodesep=5mm,offset=0pt]A2){\pnode{I2}}
\ncline{->}{I2}{A2}
\rput([angle=90,nodesep=5mm,offset=0pt]A5){\pnode{I5}}
\ncline{->}{I5}{A5}
\ncline{->}{A1}{A4}\aput[0pt](.25){\scriptsize $x,y$}
\ncline{->}{A2}{A4}\bput[0pt](.25){\scriptsize $x,y$}
\ncline{->}{A3}{A1}\aput[0.5pt](.5){\scriptsize $x$}
\ncline{->}{A3}{A2}\aput[0.5pt](.5){\scriptsize $x$}
\ncline{->}{A4}{A3}\aput[0.5pt](.5){\scriptsize $x,y$}
\ncline{->}{A4}{A5}\bput[0.5pt](.5){\scriptsize $x,y$}
\ncline{->}{A5}{A1}\bput[0.5pt](.5){\scriptsize $x$}
\ncline{->}{A5}{A2}\bput[0.5pt](.5){\scriptsize $x$}
\ncarc[arcangle=12]{<-}{A1}{A2}\aput[1pt](.5){\scriptsize $x,y$}
\ncarc[arcangle=12]{<-}{A2}{A1}\aput[1pt](.5){\scriptsize $x,y$}
\ncarc[arcangle=12]{<-}{A5}{A3}\aput[1pt](.5){\scriptsize $y$}
\ncarc[arcangle=12]{<-}{A3}{A5}\aput[1pt](.5){\scriptsize $y$}
\nccurve[angleA=195,angleB=285,ncurv=4]{->}{A1}{A1}\bput[0.5pt](.50){\scriptsize $x,y$}
\nccurve[angleA=-15,angleB=-105,ncurv=4]{->}{A2}{A2}\aput[0.5pt](.50){\scriptsize $x,y$}
\nccurve[angleA=63,angleB=-27,ncurv=4]{->}{A3}{A3}\aput[0.5pt](.50){\scriptsize $y$}
\nccurve[angleA=45,angleB=135,ncurv=4]{->}{A4}{A4}\bput[1pt](.50){\scriptsize $x$}
\nccurve[angleA=117,angleB=207,ncurv=4]{->}{A5}{A5}\bput[0.5pt](.50){\scriptsize $y$}
\NormalCoor
%%%%%%%%%%%
\pnode(6,-1){C2}
\SpecialCoor
\rput(C2){\cnode[doubleline=true]{3.5mm}{GSE}}
\rput(GSE){\scriptsize$\sigma_\varepsilon$}
\rput([angle=180,nodesep=5mm,offset=0pt]GSE){\pnode{GI}}
\ncline{->}{GI}{GSE}
\rput([angle=135,nodesep=15mm,offset=0pt]GSE){\cnode{3mm}{GSX}}
\rput(GSX){\scriptsize$\sigma_x$}
\rput([angle=45,nodesep=15mm,offset=0pt]GSE){\cnode[doubleline=true]{3.5mm}{GSY}}
\rput(GSY){\scriptsize$\sigma_y$}
\ncline{->}{GSE}{GSX}\aput[1pt](.5){\scriptsize $x$}
\ncline{->}{GSX}{GSY}\aput[1pt](.5){\scriptsize $x,y$}
\ncline{->}{GSE}{GSY}\bput[1pt](.5){\scriptsize $y$}
\nccurve[angleA=90,angleB=0,ncurv=4]{->}{GSY}{GSY}\aput[1pt](.50){\scriptsize $x,y$}
\NormalCoor
%%%%%%%%%%%
\end{pspicture}
%%%%
\end{center}
\caption{\scriptsize The transition graph of the fuzzy automaton from Example \ref{ex:red.no.det} (a)), and its Nerode automaton ${\cal A}_N$ (b)).}\label{fig:red.no.det}
%%%%
\end{figure}
%%%%
\end{example}

\begin{example}\label{ex:inf}\rm
Let ${\cal A}$ be an automaton over the one-element alphabet $X=\{x\}$ and the Goguen (product) structure given by the transition graph shown in Fig.~\ref{fig:inf} a).
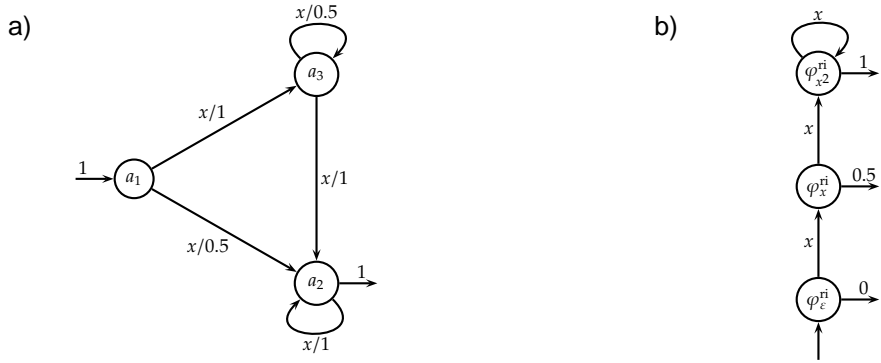
\begin{figure}
%%%%
\begin{center}
%%%%
\psset{unit=1cm}
\newpsobject{showgrid}{psgrid}{subgriddiv=1,griddots=10,gridlabels=6pt}
%%%
\begin{pspicture}(-6,-0.1)(11,4)%\showgrid
\rput(-4.5,4){\textsf{a)}}
\rput(4,4){\textsf{b)}}
\pnode(-3,2){AP1}
\SpecialCoor
\rput(AP1){\cnode{2.7mm}{AP1A1}}
\rput(AP1A1){\scriptsize$a_1$}
\rput([angle=-30,nodesep=25mm,offset=0pt]AP1A1){\cnode{3mm}{AP1A2}}
\rput(AP1A2){\scriptsize$a_2$}
\rput([angle=30,nodesep=25mm,offset=0pt]AP1A1){\cnode{3mm}{AP1A3}}
\rput(AP1A3){\scriptsize$a_3$}
\rput([angle=180,nodesep=5mm,offset=0pt]AP1A1){\pnode{AP1I}}
\ncline{->}{AP1I}{AP1A1}\aput[1pt](.2){\scriptsize $1$}
\rput([angle=0,nodesep=5mm,offset=0pt]AP1A2){\pnode{AP2O}}
\ncline{->}{AP1A2}{AP2O}\aput[1pt](.6){\scriptsize $1$}
\ncline{<-}{AP1A2}{AP1A1}\aput[1pt](.5){\scriptsize $x/0.5$}
\ncline{->}{AP1A1}{AP1A3}\aput[1pt](.5){\scriptsize $x/1$}
\ncline{->}{AP1A3}{AP1A2}\aput[1pt](.5){\scriptsize $x/1$}
\nccurve[angleA=135,angleB=45,ncurv=4]{->}{AP1A3}{AP1A3}\aput[0.5pt](.50){\scriptsize $x/0.5$}
\nccurve[angleA=-45,angleB=-135,ncurv=4]{->}{AP1A2}{AP1A2}\aput[0.5pt](.50){\scriptsize $x/1$}
\NormalCoor
%%%%%%%%%%%
\pnode(6,0.4){C2}
\SpecialCoor
\rput(C2){\cnode{3mm}{GSE}}
\rput(GSE){\scriptsize$\varphi^{\textrm{ri}}_\varepsilon$}
\rput([angle=-90,nodesep=5mm,offset=0pt]GSE){\pnode{GI}}
\ncline{->}{GI}{GSE}
\rput([angle=0,nodesep=5mm,offset=0pt]GSE){\pnode{GSEO}}
\ncline{->}{GSE}{GSEO}\aput[1pt](.6){\scriptsize $0$}
\rput([angle=90,nodesep=12mm,offset=0pt]GSE){\cnode{3mm}{GSX}}
\rput(GSX){\scriptsize$\varphi^{\textrm{ri}}_{x}$}
\rput([angle=0,nodesep=5mm,offset=0pt]GSX){\pnode{GSXO}}
\ncline{->}{GSX}{GSXO}\aput[1pt](.6){\scriptsize $0.5$}
\ncline{->}{GSE}{GSX}\aput[1pt](.5){\scriptsize $x$}
\rput([angle=90,nodesep=12mm,offset=0pt]GSX){\cnode{3mm}{GSX2}}
\rput(GSX2){\scriptsize$\varphi^{\textrm{ri}}_{\!x^2}$}
\ncline{->}{GSX}{GSX2}\aput[1pt](.5){\scriptsize $x$}
\rput([angle=0,nodesep=5mm,offset=0pt]GSX2){\pnode{GSX2O}}
\ncline{->}{GSX2}{GSX2O}\aput[1pt](.6){\scriptsize $1$}
\nccurve[angleA=135,angleB=45,ncurv=4]{->}{GSX2}{GSX2}\aput[0.5pt](.50){\scriptsize $x$}
\NormalCoor
%%%%%%%%%%%
\end{pspicture}
%%%%
\end{center}
\caption{\scriptsize Transition graphs of the fuzzy automaton ${\cal A}$ from Example \ref{ex:inf} (a)), and the automaton ${\cal A}_{\varphi^{\textrm{ri}}}$ (b)).}\label{fig:inf}
%%%%
\end{figure}
%%%%

It is easy to check that $\sigma_\varepsilon=[\,1\ 0\ 0\,]$, $\sigma_x=[\,0\ 0.5\ 1\,]$,
and $\sigma_{x^n}=[\,0\ 1\ 0.5^{n-1}\,]$, for each $n\in \Bbb N$, $n\geqslant
2$.~Therefore, the Nerode automaton of $\cal A$ has infinitely many states.

On the other hand, using Algorithms \ref{alg:gri} and \ref{alg:gwri} we obtain that
\[
\varphi^{\textrm{ri}}=\varphi^{\textrm{wri}}=\begin{bmatrix}
1 & 0 & 0.5 \\
1 & 1 & 1 \\
1 & 0 & 1
\end{bmatrix},
\]
and we construct the automaton ${\cal A}_{\varphi^{\textrm{ri}}}$ which is
shown in Fig.~\ref{fig:inf} b).

Therefore, although the Nerode automaton of $\cal A$ is inifinite, using
the greatest right invariant fuzzy quasi-order on $\cal A$ we obtain a finite crisp-deterministic
fuzzy automaton which is equivalent to $\cal A$.

\end{example}

\end{document}